\documentclass{amsart}
\usepackage{tabularx,amsmath,amsfonts,amssymb,boxedminipage,graphicx,mathtools,appendix,physics, fullpage, amsthm}
\usepackage{epstopdf, float}
\usepackage{algorithm2e}
\usepackage{subfigure}
\usepackage[final]{hyperref}
\usepackage{csquotes}
\MakeOuterQuote{"}
\RestyleAlgo{boxruled}
\hypersetup{
	colorlinks=true,       
	linkcolor=blue,          
	citecolor=blue,        
	filecolor=magenta,      
	urlcolor=blue         
}
\theoremstyle{plain} 
\newtheorem{theorem}{Theorem}[section]

\theoremstyle{definition} 

\title{Computer Assisted Discovery of Integrability via SILO:  Sparse Identification of Lax Operators}

\author{
    Jimmie Adriazola$^{1,*}$, 
    Wei Zhu$^{2}$, 
    Panayotis G. Kevrekidis$^{3}$,
    Alejandro Aceves$^{4}$
}

\begin{document}

\maketitle

\begin{center}
\small
$^1$ School of Mathematical Sciences and Statistics, Arizona State University \\
$^2$ School of Mathematics, Georgia Institute of Technology\\
$^3$ Department of Mathematics and Statistics, University of Massachussets, Amherst. \\
$^4$ Department of Mathematics, Southern Methodist University. \\
$^*$ Corresponding author, \href{mailto:jimmie.adriazola@asu.edu}{jimmie.adriazola@asu.edu}.
\end{center}

\begin{abstract}
    We formulate the discovery of Lax integrability of Hamiltonian dynamical systems as a symbolic regression problem, which, loosely speaking, seeks to maximize the compatibility between a pair of Lax operators and the known Hamiltonian of the dynamical system. 
    Our approach is first tested on the simple harmonic oscillator. We then move on to the Henon-Heiles system, i.e. a two-degree-of-freedom system of nonlinear oscillators. The integrability of the Henon-Heiles system is critically dependent on a set of three parameters within its Hamiltonian, a fact that we leverage to assess the robustness of our approach in detecting the integrability of this system with respect to the parameter dependence of the Hamiltonian. We then adapt our method to canonical examples of Hamiltonian partial differential equations, including the Korteweg-de Vries and cubic nonlinear Schr\"odinger equations, again testing robustness against nonintegrable perturbations of their respective Hamiltonians. In all examples, our approach reliably confirms or denies the integrability of the equations of interest. Moreover, by appropriately adjusting the loss function and applying thresholded $l^0$ regularization to enforce sparsity in the operator weights, we successfully recover accurate forms of the Lax pairs despite wide initial hypotheses on the operators. Some of the relevant Lax pairs, notably for the Henon-Heiles system and the Korteweg-deVries equation, are distinct from the ones that are typically reported in the literature. The Lax pairs that our methodology discovers warrant further mathematical and computational investigation, and we discuss extensively the opportunities for further improvement of SILO as a viable tool for interpretable exploration of integrable Hamiltonian dynamical systems.
    
\end{abstract}



\section{Introduction}

The theory of integrable dynamical systems occupies a central place in the historical development of mathematical physics~\cite{o2008integrable,ablowitz1981solitons,newell1985solitons,Ablowitz2011a}, further inspiring branches of modern mathematical fields in differential geometry~\cite{guillemin1990symplectic, bolsinov2004integrable}, algebra~\cite{drinfel1988quantum,newell1985solitons}, and functional analysis~\cite{deift1993steepest,zhou1989riemann}. The classical sense of the term integrability refers to the exact solvability of a differential equation up to the calculation of a suitable number of integrals~\cite{lagrange1853mecanique,poincare1967new}. More precisely, according to the Liouville-Arnold theorem, if an $n$ degree of freedom Hamiltonian dynamical system has $n$ algebraically independent integrals of motion, whose level sets are compact, then there exists a local coordinate transformation to a so-called action-angle torus~\cite{arnol2013mathematical}. To resolve the dynamics, one only needs to perform the necessary quadratures in this action-angle coordinate system~\cite{goldstein2002classical}. Moreover, since the dynamics of an integrable system is constrained to take place on such tori, the dynamics cannot exhibit chaos and is thus regular and indeed generically quasiperiodic in its nature~\cite{ott2002chaos}. Of course, there are extended notions of integrability in various different mathematical contexts; see for a recent discussion, e.g.,~\cite{Liu_2022}. Yet, what we have just described is the Liouville sense of integrability that we will often adopt throughout this paper when we refer to the integrability of a dynamical system.

If a Hamiltonian system is integrable, then Hamilton-Jacobi theory~\cite{jose1998classical} or, notably,  inverse scattering transform (IST) methods~\cite{ablowitz1981solitons,newell1985solitons,Ablowitz2011a} may be employed to construct these action-angle coordinates and ultimately solve the mechanics problem. 
The IST amounts effectively to a nonlinear form of the
Fourier transform in which suitable scattering data
are propagated forward in time and eventually the solution
of an inverse problem reconstructs the solution of the PDE
at future times.

However, one does not know a priori if a given Hamiltonian dynamical system is integrable, a feature that
would be highly desirable to recognize. Newton knew that the two-body gravitational problem is integrable; he solved the problem exactly~\cite{goldstein2002classical}. 
Yet, while attacking the three-body problem in the late seventeenth century, Newton, when pressed by his colleague Edmund Halley, asserted that the problem~\cite{westfall1980never} ``made his headache, and kept him awake so often, that he would think of it no more." 
This problem later also challenged Poincar{\'e},
as has been eloquently described, e.g., in~\cite{DiacuHolmes1996}. In the end, he correctly asserted the nonintegrability of the three-body problem. A mathematical and historical perspective of Poincar{\'e}'s work  can be found in \cite{barrow}.

Thus, the efforts of Newton and all those who attempted to solve the three-body problem exactly were always in vain. Notwithstanding the beautiful mathematics that resulted from such efforts, had nonintegrability been confirmed at the time, it would have spared the celestial mechanics community throughout the scientific revolution great pain. Considering that there will always exist a vast number of physical models whose integrability is left to be determined by mathematicians and physicists, a tool for the automated detection of integrability retains its value even in the present day.  

Founded on this principle, our aim in this paper is to develop and demonstrate a broad-scope detector of integrability for systems with a known Hamiltonian.  The key technology we use in the detection of integrability is a symbolic regression~\cite{vladislavleva2020symbolic} over a space of hypothesized Lax pairs~\cite{lax1968integrals}. With this tool that we call the Sparse Identification of Lax Operators (SILO), we expect to catalyze the further future study of Hamiltonian systems where the integrability of the dynamical system is unknown. We note in passing here that there exist substantial efforts
during the past few years towards identifying that Hamiltonian
directly from data~\cite{bertalan2019learning,CHEN2025108563, greydanus2019hamiltonianneuralnetworks,Liu_2022,tothhamiltonian}. Accordingly, the discovery of integrability directly
from data (rather than from equations) is a natural next step along the vein presented herein.

To keep this high-level and introductory discussion simple, let us consider a mechanical system with finite degrees of freedom $n$. Lax pairs in this setting are time-dependent matrices  $L(t)$ and  $P(t)$ in $\mathbb{R}^{n\times n}$ that satisfy the matrix equation
\begin{equation}\label{eq:Lax}
\frac{d L}{d t}=[L, P]
\end{equation}
where $[L,P]=LP-PL$ denotes the matrix commutator. It is a fact that any finite-dimensional $L$, where each component $L_{ij}\in C^1([0,T])$ with $t\in[0,T],$ that satisfies Equation~\eqref{eq:Lax} is isospectral~\cite{lax1968integrals,krishnaswami2020introduction}. 

Now, recall that a Hamiltonian system in canonical coordinates, with a smooth Hamiltonian function $H:\mathbb{R}^{2n}\times[0,T]\to\mathbb{R},$ is given by 
\begin{equation}\label{eq:HamEq}
\begin{gathered}
        \dot{q}_i = \frac{\partial H}{\partial p_i},\qquad\ i=1,...,n,\\
        \dot{p}_i =-\frac{\partial H}{\partial q_i},\qquad\ i=1,...,n.
        \end{gathered}
\end{equation}
If Equations~\eqref{eq:HamEq} further arise as the compatibility conditions of Equation~\eqref{eq:Lax}, then this Hamiltonian dynamical system is said to be  \textit{Lax} integrable, with integrals of motion
that can be constructed from the associated Lax operator $L$~\cite{krishnaswami2020introduction}, assuming that the resulting integrals are in involution. If it can be shown that this Lax pair reproduces all conserved quantities or satisfies certain geometric conditions, then the Lax integrability implies Liouville integrability~\cite{gui1989liouville}.  Nevertheless, the presence of a Lax pair underlying the integrable dynamics enables us to reformulate the detection of integrability in the present work. This will be formulated and accomplished as a symbolic regression problem that seeks to maximize the compatibility [in the precise sense of Equation~\eqref{eq:pointwise_relation} discussed below] between a hypothesized library of Lax pairs and a known Hamiltonian of the dynamics. This is the essence of our work.

We are aware of two previous attempts that use Lax pairs for numerical detection of integrability, as well as a third, different approach towards algorithmic discovery of integrability. The first is due to Krippendorf, et al.~\cite{krippendorf2021integrability}, where neural networks are employed to find parametrized Lax pairs that are as compatible with the known Hamiltonian dynamical system as possible. Although this work was successful in its approach and, indeed, paved the way toward such data-driven methodologies for detecting integrability, its effectiveness was somewhat limited. Due to the use of neural networks, the numerical precision and interpretability of their results
could, arguably, be further improved. Additionally, the sampling scheme used for their PDE problems assumes that one knows the soliton solution of the model.  Koster and Wahls took a different, more sophisticated data-driven approach to construct the spectral operator for members of the Ablowitz-Kaup-Newell-Segur (AKNS) hierarchy, even discovering the governing equations from noisy data~\cite{kw}. However, this approach is inherently restricted, as it is only applied to AKNS systems
(and thus uses rather limited libraries of functions).  Moreover, this approach uses an error function that requires a priori
knowledge (and necessitates the subsequent in time
preservation) of the conserved quantities of the system.

A third, very recent approach, developed by Kantamneni, Liu, and Tegmark~\cite{kantamneni2024optpde}, attempts to search over the space of possible PDEs to maximize the number of conserved quantities. However, this approach is limited in that it can only discover a few conserved quantities at most
(instead of an algorithmic sequence of infinitely many such quantities). Thus, it subsequently relies on human judgment to determine whether or not the hierarchy of conserved quantities can be continued onward towards an infinite number of such (which, moreover,
need to be in involution).

We find that reinterpreting the problem of detecting integrability via Lax pairs as a sparse regression problem yields highly promising results.
First, this allows us to obtain high-precision detection of integrability against non-integrable Hamiltonian perturbations. Second, by using sparse regression techniques, we are able to rediscover accurate structural forms of the Lax pairs in numerous examples.
These include ones that are accurate even though they were not ----and
in some cases, to the best of our knowledge, are not------ the
standard ones used in the literature. We note that the idea of searching over a wide library of hypothesized operators while maintaining a sparsification is similar in spirit to the Sparse Indentification of Nonlinear Dynamics (SINDy) methodology that discovers governing equations from data~\cite{bpk} . A major departure from SINDy here is largely in the execution of the computations, since discovery of sparse Lax pairs from data involves an entirely different loss function than what is used in equation discovery.

Consequently, our work demonstrates that SILO advances us in the direction of building general, robust, interpretable, and reliable tools for precise mathematical discoveries in the context of integrable Hamiltonian systems. We also aim to be as transparent as possible about our methodology in that we discuss at length the limitations, computational choices, and many of the salient implementation details when using SILO. We reiterate here what we consider to be a major advantage of framing integrability detection via a sparse regression on Lax pairs, as opposed to a numerical search for models with a high, yet finite number of conserved quantities. 
That is, once Lax integrability is established,  work by Gui-zhang~\cite{gui1989liouville} demonstrates that one can, in principle, elevate this to Liouville integrability by checking certain geometric conditions on the discovered Lax pair. Thus, the discovery of \textit{Liouville} integrability from Lax integrability remains, in principle, automated with the use of computer algebra systems. Such an examination goes beyond the scope of our paper, but we make note of this here to communicate that such studies can be undertaken if desired.

The paper is organized as follows. In Section~\ref{section:SHO}, we build and motivate the loss function, used in the regression, for a system with one degree of freedom. This loss function, as we will show, scales to scenarios with many degrees of freedom. In this section, we also take the opportunity to detail our sparsification strategy, employed throughout this work, as much as possible. Since autonomous Hamltonian systems with one degree of freedom are always integrable, we proceed to study a two-degree-of-freedom system, the Henon-Heiles model, in Section~\ref{section:HH}. Because the integrability of the Henon-Heiles model depends on parameters defining its Hamiltonian, this case allows us to assess the accuracy of our method for detecting integrability. In Section~\ref{section:KdV} and Section~\ref{section:NLS}, we extend and test our method for Hamiltonian nonlinear partial differential equations (PDEs) including the Korteweg-de Vries equation and the nonlinear Schr\"odinger equation, respectively. In all sections, we employ a sparsification technique for the regression to facilitate the interpretation of the Lax pair discovered by our numerics. We also comment on the relevant findings in comparison
to the ``standard'' Lax pair settings that exist in the corresponding theoretical literature.

\section{Warm Up Problem: The Simple Harmonic Oscillator}\label{section:SHO}

In this section, we introduce the basic methodology we use to detect integrability for Hamiltonian systems and interpret the resulting Lax pairs. As an instructive warm-up, we first develop the method for the simple harmonic oscillator. We find that with appropriate adjustments, the techniques discussed in this section can be scaled to the discovery of Lax Pairs for Hamiltonian PDEs, i.e. systems with
infinitely many degrees of freedom. We emphasize again that while certain aspects of our methodology share similarities with those of Krippendorf et al.~\cite{krippendorf2021integrability}, our method differs significantly by adopting a more fundamental interpretation of the optimization problem formulated. 
Specifically, the approach presented here uses
an expanded basis and a more robust sampling approach compared to the earlier work
of~\cite{krippendorf2021integrability}. Moreover,  the subsequent sparsification process leads to the interpretable Lax pairs identified
throughout this work. Furthermore, our work only uses
the known Hamiltonian and is, in principle, agnostic to knowledge of conserved quantities or mathematical form of the Lax pair, such as what is assumed in~\cite{kw}.

To begin, recall that the dynamics of the simple harmonic oscillator are given by the following system,
\begin{equation}\label{eq:SHO}
\begin{gathered}
        \dot{q} = \frac{\partial H}{\partial p}=p/m ,\\
        \dot{p} =-\frac{\partial H}{\partial q}= - k q ,
\end{gathered}
\end{equation}
where the Hamiltonian $H=\left(p^2/m+m \omega^2 q^2\right)/2$ and the spring constant $k$ is defined in terms of the oscillator's mass $m>0$ and angular frequency $\omega\in\mathbb{R}$, namely $k=m \omega^2$.

A Lax pair reported by the literature for this system is given by~\cite{krishnaswami2020introduction}
\begin{equation}\label{eq:SHOLaxPair}
L=\left(\begin{array}{cc}
p / m & \omega q \\
\omega q & -p / m
\end{array}\right) \quad \text { and } \quad P=\left(\begin{array}{cc}
0 & \omega / 2 \\
-\omega / 2 & 0
\end{array}\right).
\end{equation}
It is important to note that the Lax pair~\eqref{eq:SHOLaxPair} corresponding to Equation~\eqref{eq:SHO} is not unique; indeed, Krippendorf, et al.~\cite{krippendorf2021integrability} correctly point out that there exist two two-parameter families of Lax pairs
\begin{subequations}
\begin{equation}\label{eq:FirstSHOLax}
 L_1=a\left(\begin{array}{cc}
p/m & b \omega q \\
\omega q/b & -p/m
\end{array}\right), P_1=\left(\begin{array}{cc}
0 & b\omega/2 \\
-\omega/2b & 0
\end{array}\right),
\end{equation}
\begin{equation}\label{eq:SecondSHOLax}
 L_2=a\left(\begin{array}{cc}
q & p/b\omega \\
bp/\omega & -q
\end{array}\right), P_2=\left(\begin{array}{cc}
0 & -\omega/2b \\
2\omega/b & 0
\end{array}\right),
\end{equation}
\end{subequations}
where $a,b\in\mathbb{R}/\{0\}.$ 
Also note that the only conserved quantity, the energy of the system $E$, is related to the eigenvalues of the operator $L$ appearing in~\eqref{eq:SHOLaxPair} and can be written as
\begin{equation}
	E=\frac{1}{2}\left(\frac{p^2}{m}+m \omega^2 q^2\right)=-\frac{m}{2} \operatorname{det} L=\frac{m}{4} \operatorname{tr} L^2.
\end{equation}

Already at this elementary stage, the nonuniqueness of the Lax pairs prevents us from ever making the claim that we can reconstruct \textit{the} Lax pair for a Hamiltonian dynamical system. Nevertheless, we will attempt to construct one such pair. We hypothesize the operators $\tilde{L}$ and $\tilde{P}$, satisfying Lax's equation, that we would like construct, are of the form
\begin{equation*}
\begin{gathered}
        \tilde{L}(t) = \sum_{k=1}^{N_{\xi}} \xi_k \Theta_k^{(L)} (t),\\
    \tilde{P}(t) = \sum_{k=1}^{N_{\zeta}} \zeta_k \Theta_k^{(P)}(t),
    \end{gathered}
    \end{equation*}
where the symbol $\Theta$ represents a user-specified library of matrices with parameters $\xi\in\mathbb{R}^{N_{\xi}}$ and $\zeta\in\mathbb{R}^{N_{\zeta}}$. In the case of the simple harmonic oscillator, we hypothesize that the $2\times 2$ matrix operators contain terms that are at most linear in both $q$ and $p$. This gives $N_{\xi}=N_{\zeta}=2\times2\times3=12.$ We make no other assumptions on the form of the operators $\tilde{L}$ and $\tilde{P}$. Note that despite limiting each entry to a linear polynomial, we already have 24 parameters $\eta:=[\xi\ \zeta]$ to search for.

It is obvious that the size of this parameter space $N_{\eta}$ is much too large; in fact,  by at least a factor of 6.  Indeed, if we make the following assumptions
\begin{itemize}
	\item The operator $\tilde{L}$ does not contain any constant terms,
	\item The operator $\tilde{L}$ is a symmetric matrix with zero trace,
	\item The operator $\tilde{P}$ is skew-symmetric,
\end{itemize}
then there are only 4 parameters to find. To show this, we write the hypotheses explicitly:
$$
\tilde{L}= 	\begin{pmatrix}
	\xi_1 q + \xi_3 p  & \xi_2 q +\xi_4 p \\
	\xi_2 q +\xi_4 p  & -\xi_1 q - \xi_3 p
\end{pmatrix},
$$
and
$$
\tilde{P}= 	\begin{pmatrix}
	0  & \zeta_1 x +\zeta_2 p +\zeta_3\\
	-\zeta_1 q -\zeta_2 p-\zeta_3  & 0
\end{pmatrix}.
$$
With this hypothesis, the governing equations are given by
\begin{equation}\label{eq:SmallSHOEq}
\begin{gathered}
\dot{q}=2\zeta_3\left(\frac{\xi_2}{\xi_4}-\frac{\xi_1}{\xi_3}\right)^{-1}\left(\left(\frac{\xi_1}{\xi_4}+\frac{\xi_2}{\xi_3}\right)q+\left(\frac{\xi_3}{\xi_4}+\frac{\xi_4}{\xi_3}\right)p\right), \\
\dot{p}=2\zeta_3\left(\frac{\xi_4}{\xi_2}-\frac{\xi_3}{\xi_1}\right)^{-1}\left(\left(\frac{\xi_1}{\xi_2}+\frac{\xi_2}{\xi_1}\right)q+\left(\frac{\xi_3}{\xi_2}+\frac{\xi_4}{\xi_1}\right)p\right),
\end{gathered}
\end{equation}
while
\begin{equation}\label{eq:SmallSHOTrace}
\operatorname{tr}(L^2)=2((\xi_1^2+\xi_2^2)q^2+(\xi_3^2+\xi_4^2)p^2+(\xi_1\xi_3+\xi_2\xi_4)qp).
\end{equation}
We see that any combination of coefficients satisfying $(\xi_1\xi_3+\xi_2\xi_4)=0$ eliminates the undesirable cross-terms appearing in Equation~\eqref{eq:SmallSHOTrace} and the equations of motion~\eqref{eq:SmallSHOEq}. Moreover, the parameters $\zeta_1$ and $\zeta_2$ are redundant since they do not appear anywhere in the dynamics. These 5 remaining parameters along with the condition $(\xi_1\xi_3+\xi_2\xi_4)=0$ leave 4 total parameters. 

Of course, having the hindsight to make such a judicious choice of assumptions on the Lax pairs is a luxury. When no user-specified knowledge is available, it is standard in data-driven science and statistics to sparsify the problem using some variant of the $l^p\left(\mathbb{R}^{N_{\eta}}\right)$ norm~\cite{bpk,hastie2009elements,hastie2015statistical}; an approach that we adopt in this work and will discuss in greater depth shortly.

We now motivate the loss function. The Lax equation includes a time derivative of the operator $L$. Numerically working with time derivatives naturally introduces a time discretization, leading to numerical errors. To circumvent this issue, we observe via the chain rule that
\begin{equation*}
	\begin{aligned}
		\frac{d L}{d t} & =\frac{\partial L}{\partial q} \dot{q}+\frac{\partial L}{\partial p} \dot{p} \\
		& =\frac{\partial L}{\partial q} \frac{\partial H}{\partial p}-\frac{\partial L}{\partial p} \frac{\partial H}{\partial q} \\
		& :=\{L, H\},
	\end{aligned}
\end{equation*}
where the last line defines the Poisson bracket of the operator $L$ with the Hamiltonian $H$. Evaluating this Poisson bracket $\{L,H\}$ is straightforward to do analytically by hand. Also, since $\frac{d L}{d t} =[L,P],$ if follows that
\begin{align}
\label{eq:pointwise_relation}
    \{L, H\}-[L,P]=0    
\end{align}
should hold for any $q,p\in\mathbb{R}$. However, this last equation admits the trivial solution $L\equiv0$, so care should be taken to penalize away from $\xi\equiv0$.

With these pieces in hand, we formulate the sparse identification of a Lax pair associated with a given Hamiltonian dynamical system as the following optimization problem
\begin{equation}\label{eq:SHOProb}
\min _{\eta\in\mathbb{R}^{N_{\eta}}}J[\eta]=\min _{\eta\in\mathbb{R}^{N_{\eta}}} \frac{1-r}{N_{\rm samples}}\sum_{k=1}^{N_{\rm samples}}\left[\sum_{i, j} \frac{(\{\tilde{L}, H\}-[\tilde{L}, \tilde{P}])_{i, j}^2}{\{\tilde{L}, H\}^2_{i,j}}\right]_\Omega+r\mathcal{R}(\eta),
\end{equation}
where $r\in [0, 1)$ controls the amount of desired sparsification in the search. The first contribution to the loss function $J$ rewards the discovery of the pointwise relationship~\eqref{eq:pointwise_relation}, in phase space, the Lax pair needs to satisfy. Since the numerator is divided by the Poisson bracket that acts on $\tilde{L}$, this formulation discourages trivial solutions where $\tilde{L}\equiv0$. Since it is infeasible to access the continuous nature of the phase space, we restrict the evaluation to the set $\Omega$, a subset of the phase space $\mathbb{R}^2,$ and $N_{\rm samples}$ is the number of times we sample from $\Omega$. The sum over $i,j$ ensures that we consider every component of the $2\times2$ matrices. In summary, we search for a pair of matrices that on average over some subset of the phase space minimizes the relative residual of the Lax equation, compatible with the Hamiltonian of the system, while maintaining parsimony through the sparsification function $\mathcal{R}$.

Typical choices for the sparsification function include $\mathcal{R}_0=\|{\eta}\|_{l^0\left(\mathbb{R}^{N_{\eta}}\right)}$ or $\mathcal{R}_1=\|{\eta}\|_{l^1\left(\mathbb{R}^{N_{\eta}}\right)}$. 

We find that neither of these choices works quite as well as we could hope. After several numerical experiments, we found that
\begin{equation}\label{eq:bestR}
    \mathcal{R}^*(\eta):=\frac{\|{\tilde{\eta}}\|_{l^0\left(\mathbb{R}^{N_{\eta}}\right)}}{\|{\eta}\|_{l^0\left(\mathbb{R}^{N_{\eta}}\right)}}
\end{equation}
consistently yielded the best results. Here, $\|\tilde{\eta}\|_{l^0\left(\mathbb{R}^{N_{\eta}}\right)}$ is the number of parameters that have an absolute value greater than some user-specified tolerance $\tau$. In short, the regularization $\mathcal{R}^*(\eta)$ counts the proportion of components that exist outside of the threshold level $\tau$. Thus, we always have $0<\mathcal{R}^*<1$ independent of the size $N^{\eta}$ of the search space.

Especially for the larger-dimensional problems that we tackle later in this work, the optimization problem equipped with $\mathcal{R}^*(\eta)$ is not enough to reconstruct a sparse set of Lax operators. The steps we take to perform sparsification are precisely as follows:
\begin{itemize}
	\item [1)] Solve the optimization problem~\eqref{eq:SHOProb} with some fixed $r>0$ and threshold level $\tau>0$.
	\item [2)] Solve the problem again over the subspace of nonzero parameters, using a line search method with the initial iterate set to the output of Step 1, but with no sparsification and no thresholding. That is, set $r=\tau=0.$
	\item [3)] Solve the problem yet again with $r=0$ using the output of step 2 as the initial guess for a line search, but introduce a smaller threshold level to remove any small components of $\eta$ that may still remain.
\end{itemize}

The rationale behind this approach is as follows. Having fixed the parameter $r$, Step 1 executes the sparse search without any guesswork. We used a 24-core machine to execute the search on 24 thresholding values $\tau$ in parallel. In principle, this only costs us the real-time cost of one search. In Step 2, we remove all bias that sparsification introduces into the problem by setting $r=\tau=0$. However, when we remove the threshold, this introduces the likelihood that thresholded components may be resurrected. To avoid resurrection and to reintroduce smoothness into the problem, we remove thresholding and ensure that the search is executed only in the subspace of nonzero parameter values that survived from Step 1. Now, since the search performed in Step 2 may have some remaining components that are small but nonzero, we reintroduce a threshold of $\tau/10$ for each value of $\tau$ and execute a last parallel search maintaining $r=0$ so that the sparsification bias is kept to a minimum. This strategy may not be the most elegant, but it is algorithmic and has proven successful in all the examples of this work. 

Some further comments on sparsification are warranted. For all the numerical examples in this work, we fix $r=1/2$. Varying this parameter, of course, interpolates between the two clear extremes of the optimization problem given by Equation~\eqref{eq:SHOProb}. We found, through numerical experimentation, that fixing $r$ and varying $\tau$ are sufficient for our purposes. Additionally, the more commonly used alternatives to $\mathcal{R}^*,$ such as the previously discussed $\mathcal{R}_0$ and $\mathcal{R}_1$ sparsification functions, could perhaps be made viable. One could also imagine that using $\mathcal{R}_p,$ $0<p<1$, may also be deployable. When testing these different strategies, we found that for $\mathcal{R}_0$ to be reliable, we required different values of $r$ to be used in the different examples in this paper, while $\mathcal{R}_1$ overrelaxed sparsification, and we could not discover the sparsest possible Lax pairs this way. Our sparsification strategy, as will be demonstrated throughout this work, achieves the desired result of finding sparse Lax pairs as best we could hope. For brevity, we omit detailed comparisons among different $\mathcal{R}_p,$ $0<p<1,$ and let the numerical results speak for themselves.

Returning to the simple harmonic oscillator for a numerical demonstration, we set $k=5$ and $ m=2$---chosen arbitrarily for this proof of concept. The phase space is sampled using $N_{\rm samples}=100$ points from a uniform distribution over $[-1,1]\times[-1,1].$ To numerically solve the optimization problem~\eqref{eq:SHOProb}, we use a non-convex search method called the cross-entropy method~\cite{rubinstein2005stochastic}, implemented in the CEopt MATLAB package~\cite{cunha2024ceopt}. We then feed the result of cross-entropy minimization into a BFGS method~\cite{nocedal1999numerical}, as implemented by MATLAB's \texttt{fminunc} function, which refines the result and quickly converges to the nearest local minimum. In this way, this hybrid optimization approach overcomes the non-convexity of the optimization landscape while remaining computationally efficient. Without any sparsification, we reliably minimize the loss function to a value on the order of $10^{-15}.$

For sparsification, we performed our parallel search over the 24 uniformly spaced values of $\tau\in[0.1,2.4]$. Despite the bias in the training, we are still able to achieve a loss on the order of $10^{-15}$ (evaluated when $r=0$) and correctly identify the exact structure of the known Lax pairs given by Equations~\eqref{eq:FirstSHOLax} and~\eqref{eq:SecondSHOLax}, that is, correctly identify the six parameters needed in the 24 parameter set $\eta$. The reason why 6 parameters should be identified, instead of the 4 as argued previously, is due to the free parameters $a$ and $b$ present in Equations~\eqref{eq:FirstSHOLax} and~\eqref{eq:SecondSHOLax}. 

Precisely, the algorithm discovers
$$
\tilde{L}_1= 	\begin{pmatrix}
	\eta_1 q  &\ \ \eta_4 p \\
	\eta_6 p  &\eta_7 q
\end{pmatrix},\qquad
\tilde{L}_2= 	\begin{pmatrix}
	\eta_2 p  &\ \ \eta_3 q \\
	\eta_5 q  &\eta_8 p
\end{pmatrix}
$$
and
$$
\tilde{P}= 	\begin{pmatrix}
	0  & \eta_{22}\\
	\eta_{23} & 0
\end{pmatrix},
$$
with all other 18 entries of $\eta$ equal to zero. The recovered Lax pairs reproduce the equations of motion with seven-digit precision. At this stage, it is quite encouraging that our methodology correctly identifies \textit{both} known Lax pairs. Figure~\ref{fig:SparseSHO} shows two examples of the values of $\eta$ before and after thresholding. We report that our algorithm did not find any other family of Lax pairs for the simple harmonic oscillator. We therefore report with confidence that no other such family of Lax pairs exists within the operator hypothesis used for this study. 
\begin{figure}[htbp]
	\begin{centering}
{\includegraphics[width=.45\textwidth]{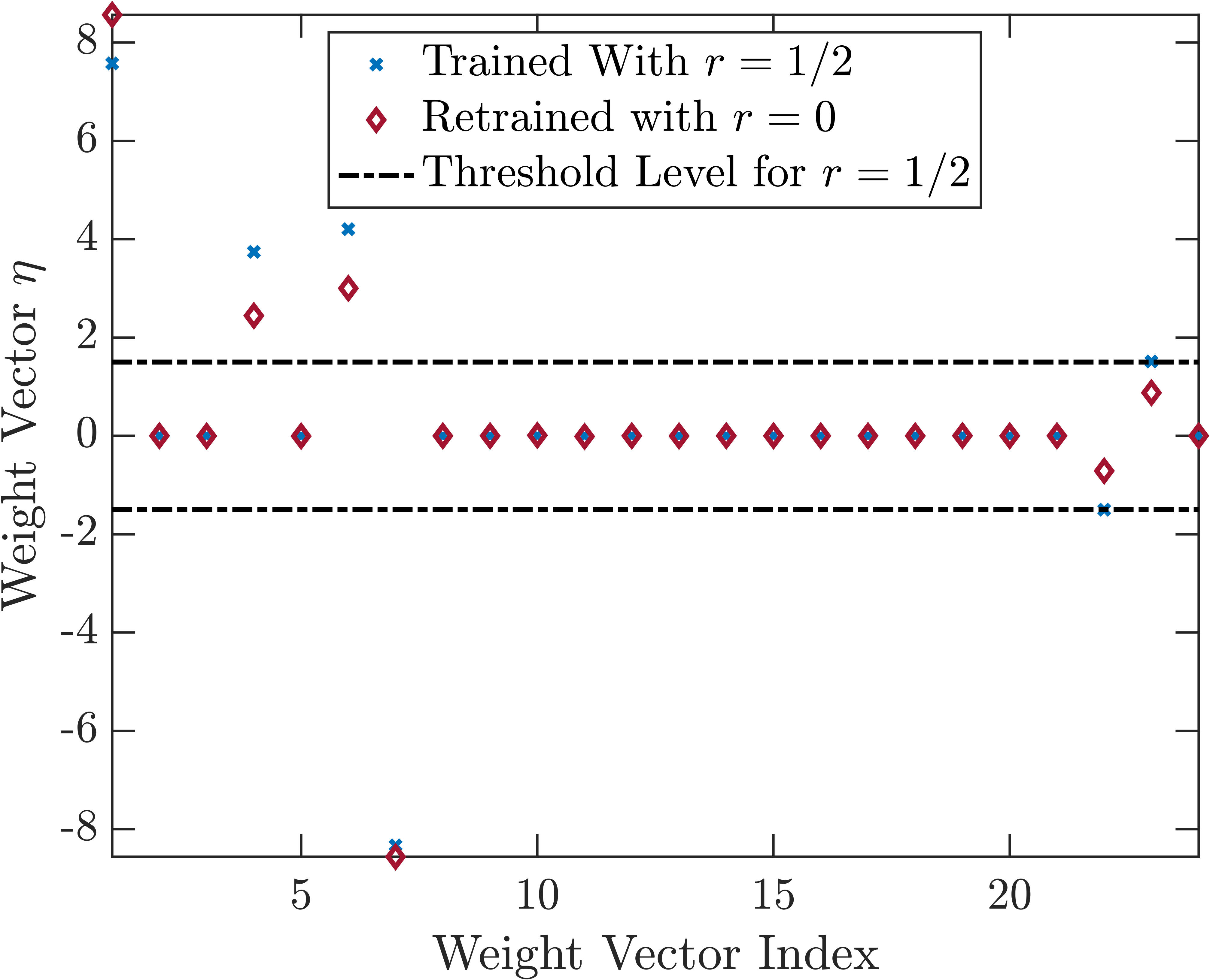}}
{\includegraphics[width=.45\textwidth]{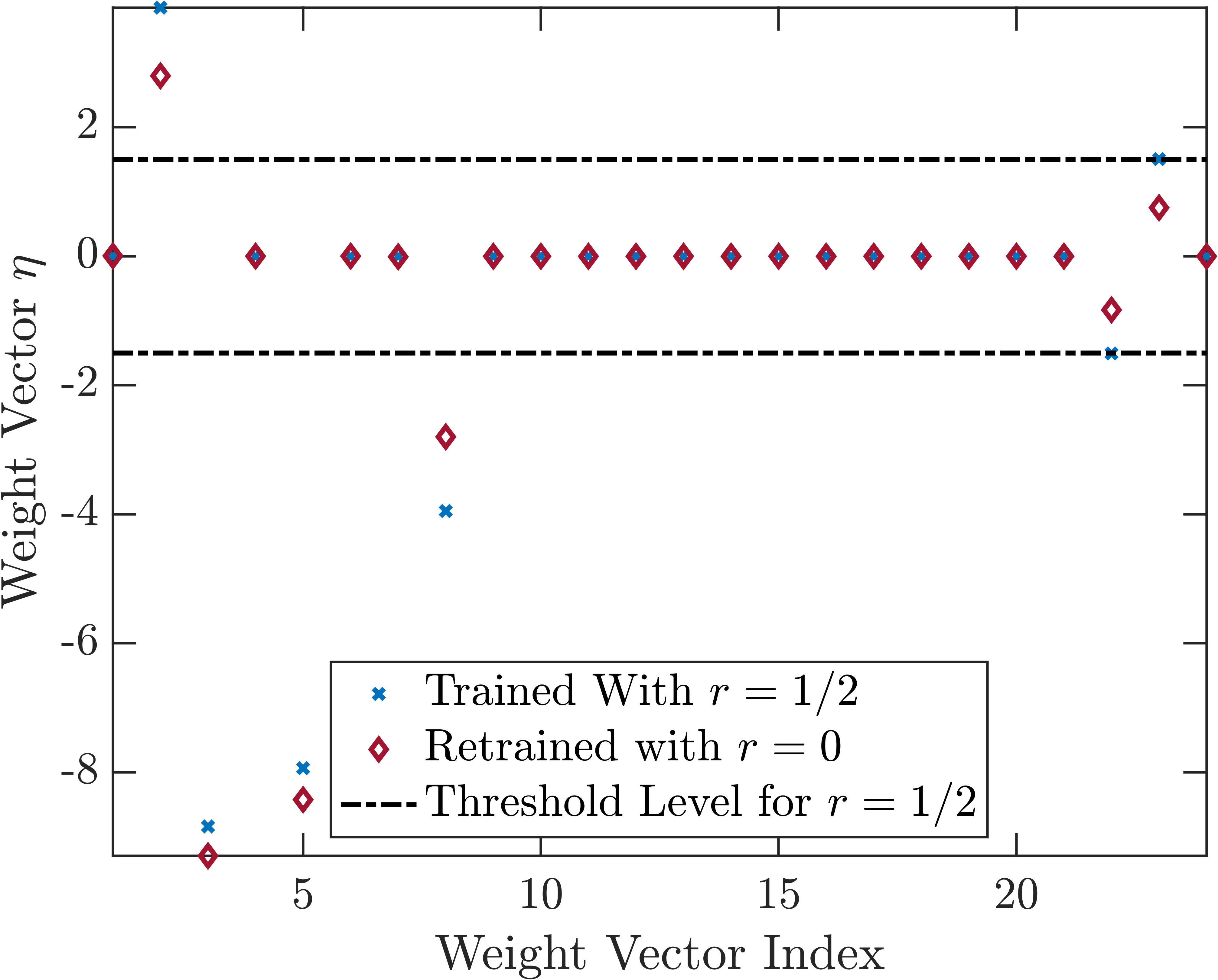}}
	\end{centering}
	\caption{Two numerical results from studying Problem~\eqref{eq:SHOProb} for the simple harmonic oscillator~\eqref{eq:SHO}. We show the result of optimizing before and after thresholding. SILO correctly identifies that there should only be 6 parameters in the 24 parameter operator hypothesis. SILO also finds both families of Lax pairs as given by Equations~\eqref{eq:FirstSHOLax} and~\eqref{eq:SecondSHOLax}.}\label{fig:SparseSHO}
\end{figure}

Thus, integrability detection and the recovery of a computationally meaningful Lax pair for the simple harmonic oscillator is achieved. But, of course, one-degree-of-freedom autonomous Hamiltonian systems are always integrable because the Hamiltonian itself is conserved. Moreover, a Lax pair can always be identified with the dynamics since for any Hamiltonian of the form
$H=\frac{1}{2}p^2+V(q)$, with $V(q)\in C^1(\mathbb{R}),$ the compatibility of the operators
\begin{equation*}
\begin{gathered}
L=\left(\begin{array}{cc}
p & \sqrt{2V(q)} \\
\sqrt{2V(q)} & -p
\end{array}\right) \\
P=\left(\begin{array}{cc}
0 & \partial_q \sqrt{2V(q)} \\
-\partial_q \sqrt{2V(q)} & 0
\end{array}\right) \\
\end{gathered}
\end{equation*}
gives rise to Hamilton's equations of motion $\dot{q}=\frac{\partial H}{\partial p},\dot{p}=-\frac{\partial H}{\partial q}.$ For this reason, we move on to a two-degree-of-freedom system, carrying forward the lessons learned from this test case.

\section{A Two Degree of Freedom Scenario: Henon-Heiles}\label{section:HH}

Now with a viable approach in hand, we go on to study a slightly more complicated example: the famous Henon-Heiles (HH) system. This is a system with two degrees of freedom defined by the Hamiltonian
\begin{equation}\label{eq:HHHam}
\begin{aligned}
H=\frac{1}{2}\left(p_x^2+p_y^2\right) & +\frac{1}{2}\left(A x^2+B y^2\right) +x^2 y+\varepsilon y^3,
\end{aligned}
\end{equation}
where the parameters $A,\ B,$ and $\varepsilon$ are arbitrary~\cite{henon1964applicability,fordy1991henon}. This system is known to be integrable for three different cases of these parameters~\cite{ravoson1993separability}. 

We will focus on one such case; $A=B$ and $\varepsilon=1/3$. In this case, a Lax pair for the HH system can be written as~\cite{ravoson1993separability}
\begin{equation}
L=\left(\begin{array}{cc}
L_1 & 0 \\
0 & L_2
\end{array}\right), \quad P=\left(\begin{array}{cc}
P_1 & 0 \\
0 & P_2
\end{array}\right)
\end{equation}
where
\begin{equation*}
\begin{aligned}
& L_1=\left(\begin{array}{cc}
p_y-p_x & y-x-k \\
\frac{2 k^2+k[3 A+2(y-x)]+3 A(y-x)+2(y-x)^2}{12} & p_x-p_y
\end{array}\right), \\
& P_1=\left(\begin{array}{cc}
0 & 1 \\
-\frac{y-x}{3}-\frac{3 A+2 k}{12} & 0
\end{array}\right), \\
& L_2=\left(\begin{array}{cc}
p_y+p_x & y+x-k \\
\frac{2 k^2+k[3 A+2(y+x)]+3 A(y+x)+2(y+x)^2}{12} & -p_x-p_y
\end{array}\right), \\
& P_2=\left(\begin{array}{cc}
0 & 1 \\
-\frac{y+x}{3}-\frac{3 A+2 k}{12} & 0
\end{array}\right),
&
\end{aligned}
\end{equation*}
and $k$ is a free parameter. 

We are interested in testing our methodology as a proof of concept. For this reason, we use a mix of basic assumptions about the structure of the Lax pair along with some sparsification. Specifically, we exploit the block-diagonal structure of the Lax pairs and inform our search so that we only seek to achieve a small loss over one of the blocks. Our hypothesis for the Lax operators is 
\begin{equation}
\tilde{L}=\left(\begin{array}{cc}
\xi_1 p_x+\xi_2 p_y & \xi_8 x+\xi_9 y+\xi_{10} x y+\xi_{11}x^2+\xi_{12}y^2 \\
\xi_3 x+\xi_4 y+\xi_5 x y+\xi_6 x^2+\xi_7 y^2 &-\xi_1 p_x-\xi_2 p_y
\end{array}\right)
\end{equation}
and 
\begin{equation}
\tilde{P}=\left(\begin{array}{cc}
0 & \zeta_4+\zeta_5 x+\zeta_6 y \\
\zeta_1+\zeta_2 x+\zeta_3 y & 0
\end{array}\right)
\end{equation}

Note that the most general quadratic hypothesis in all variables and without assuming a block diagonal structure would result in an optimization dimension of $N_{\eta}=288$. Although this remains computationally feasible, such a study would obfuscate the essence of how our methodology performs. Therefore, we instead study this much smaller problem of dimension $N_{\eta}=18$, leaving higher-dimensional studies for future work.

Only two tasks remain. First, we need the Poisson bracket $\{L,H\}$ again,
\begin{equation}
	\begin{aligned}
		\frac{d L}{d t} & =\frac{\partial L}{\partial x} \dot{x}+\frac{\partial L}{\partial y} \dot{y}+\frac{\partial L}{\partial p_x} \dot{p}_x+\frac{\partial L}{\partial p_y} \dot{p}_y \\
		& =\frac{\partial L}{\partial x} \frac{\partial H}{\partial p_x}+\frac{\partial L}{\partial y} \frac{\partial H}{\partial p_y}-\frac{\partial L}{\partial p_x} \frac{\partial H}{\partial x}-\frac{\partial L}{\partial p_y} \frac{\partial H}{\partial y} \\
		&=\{L,H\}
	\end{aligned}
\end{equation}
which is straightforward to evaluate by hand. The other is to generalize the sampling of the phase space to four dimensions, that is, $(x_k,p_{x_k},y_k,p_{y_k})\sim U([-1,1]^4)$. 

We train again on $N_{\rm samples}=100$ samples from the phase space. As illustrated in Figure~\ref{fig:SparseHH}, the optimization correctly identifies that only 13 parameters are necessary in the Lax pair. The thresholds, difficult to distinguish from zero in that scaling of the figure, were found to be on the order of $10^{-3}$ after some experimentation. We believe that such small threshold levels were needed because so much sparsification had been performed manually through the choice of our operator hypothesis. 
\begin{figure}[htbp]
	\begin{centering}		\subfigure{\includegraphics[width=.45\textwidth]{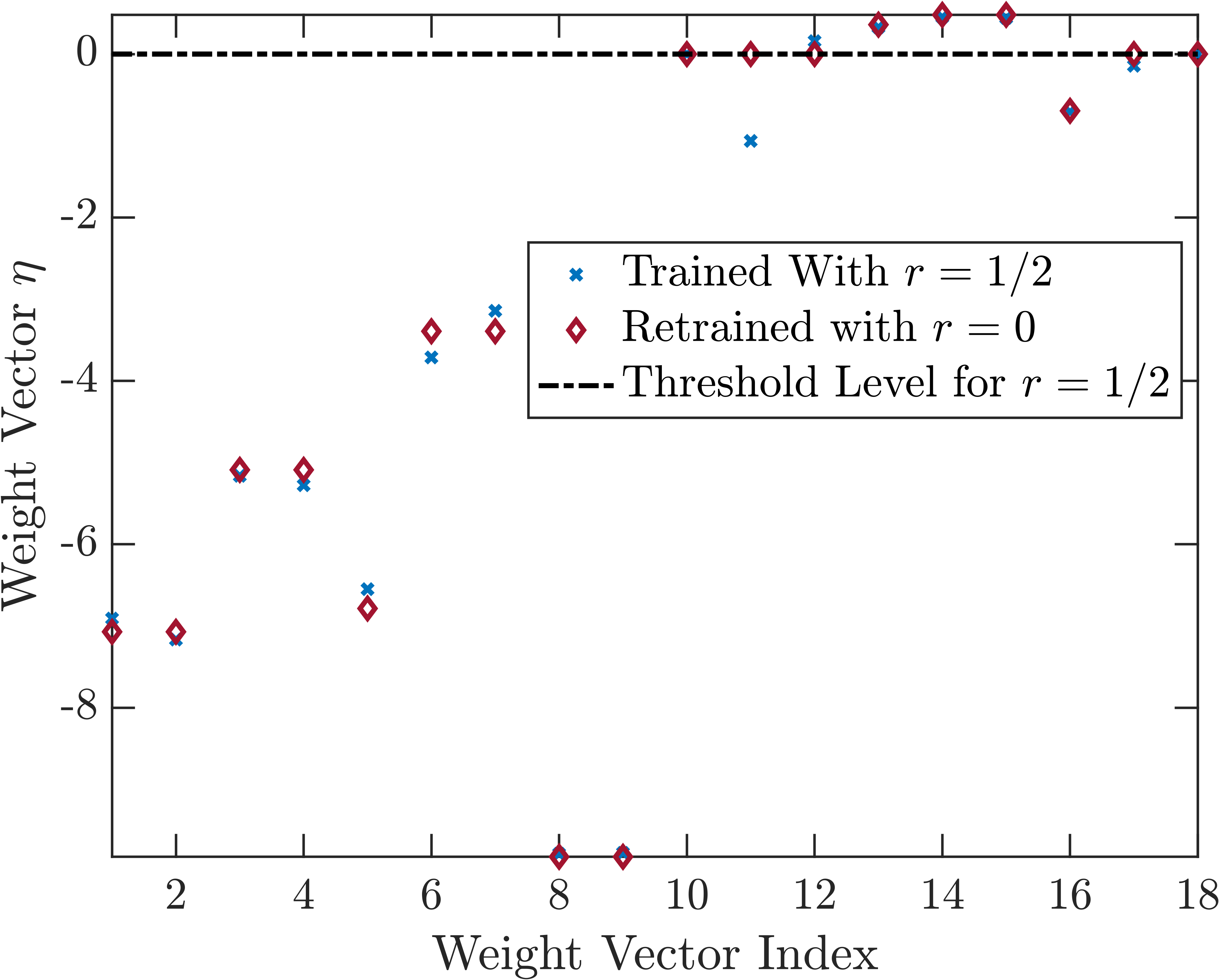}}
\subfigure{\includegraphics[width=.45\textwidth]{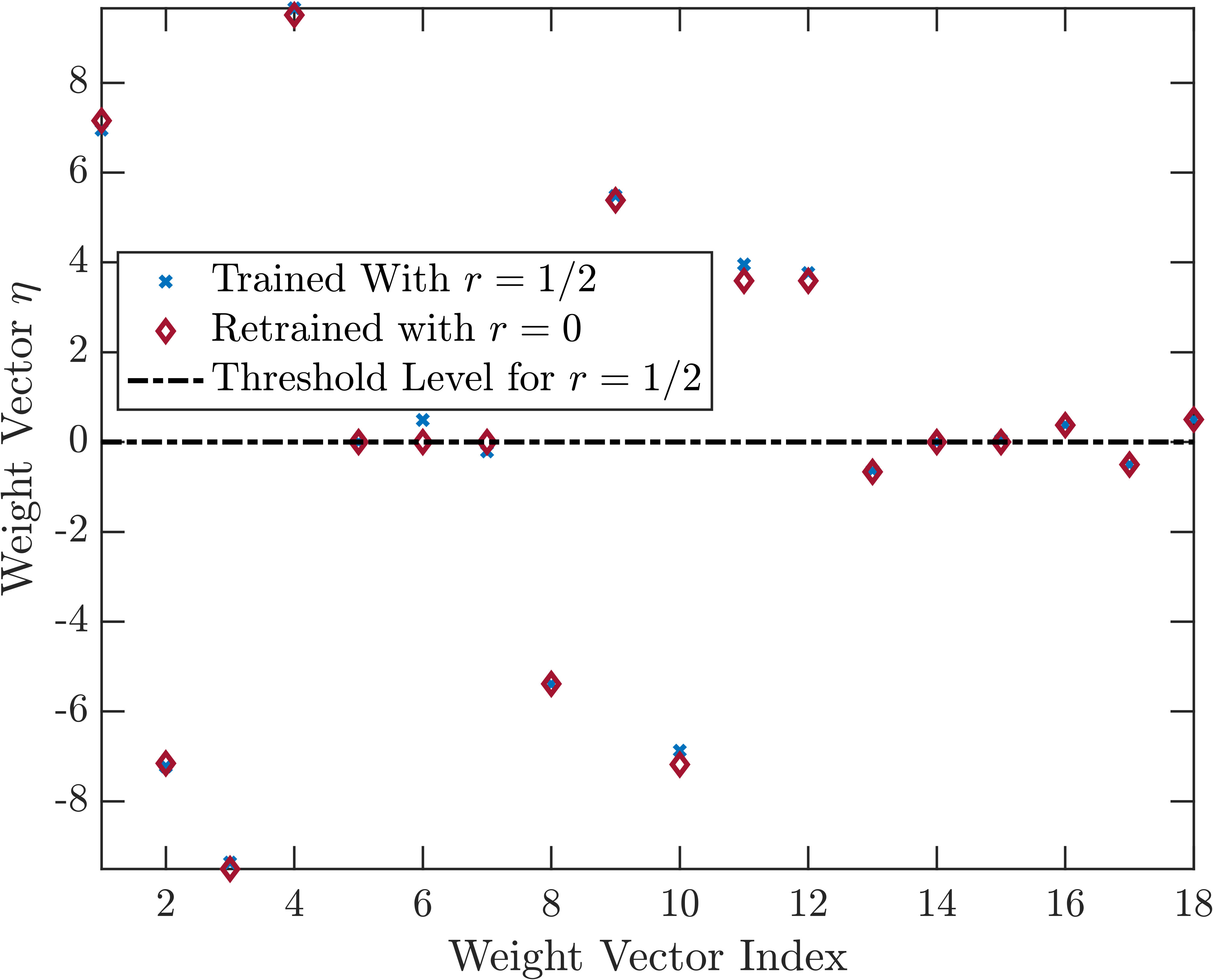}}
	\end{centering}
	\caption{The correct identification, to 6 digits of precision, of a Lax pair for the integrable case of the HH system given by Hamiltonian~\eqref{eq:HHHam}. The panel on the left reproduces the expected structure of the known Lax pair. The panel on the right discovers that, up to a sign, the transpose of the known Lax pair is also valid for reproducing the HH system.}\label{fig:SparseHH}
\end{figure}

We report that the discovered Lax pairs reproduce the dynamical system to 6 digits of precision with a loss on the order of $10^{-15}$. In fact, we recognize that one Lax pair has not been previously reported. Some simple algebra indeed shows that up to a sign, the transpose of the known Lax pair also reproduces the Henon-Heiles system. Although this finding is only mildly interesting, it once again provides further confidence that SILO, with high likelihood, indeed identifies all possible sparse options for compatible Lax pairs.

Our best results through sparsification find that the loss is a number on the order of $10^{-15}$ in the integrable case of $A=B=1$ and $\varepsilon=1/3$. We now ask, is SILO sensitive enough to detect when a system is non-integrable?  Our findings suggest that it is. We show the numerical results of the search for Lax pairs, without sparsification, as $A$ and $\varepsilon$ are varied, with $B=1$ fixed, in Figure~\ref{fig:HHdetection}. These results show a significant increase in the computed loss--by several orders of magnitude--across the parameter space. Thus, restricted to the operator hypothesis used, we can determine where the integrability lies in the parameter space $(A,\varepsilon)$ for fixed $B$. We note that sparsification is not used for this type of study because the focus is not on the interpretation of the discovered Lax pair, but instead on whether compatible Lax pairs could be identified at all.

\begin{figure}[htbp]
\begin{centering}
\subfigure{\includegraphics[width=0.47\textwidth]{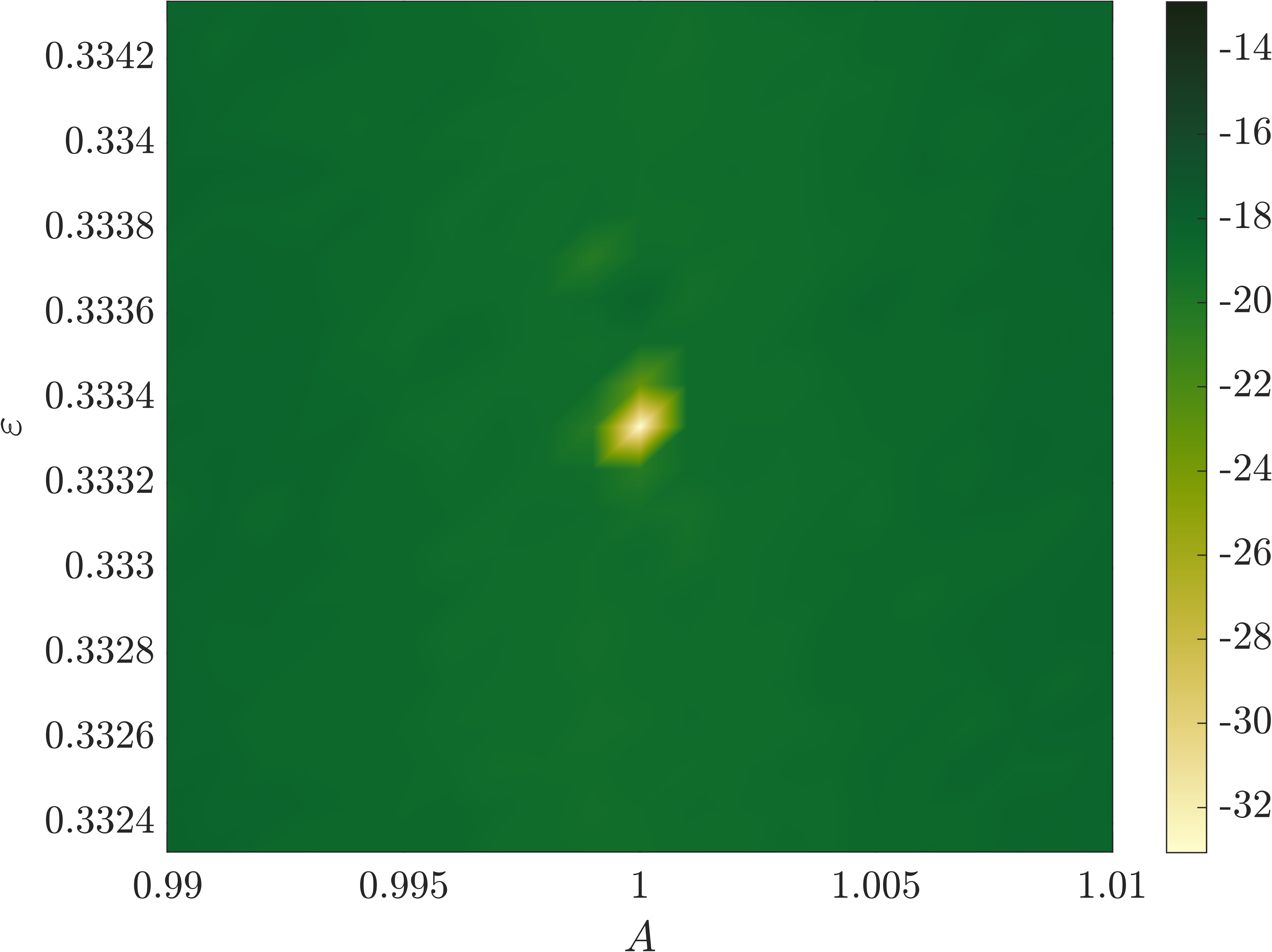}}~~
\subfigure{\includegraphics[width=0.45\textwidth]{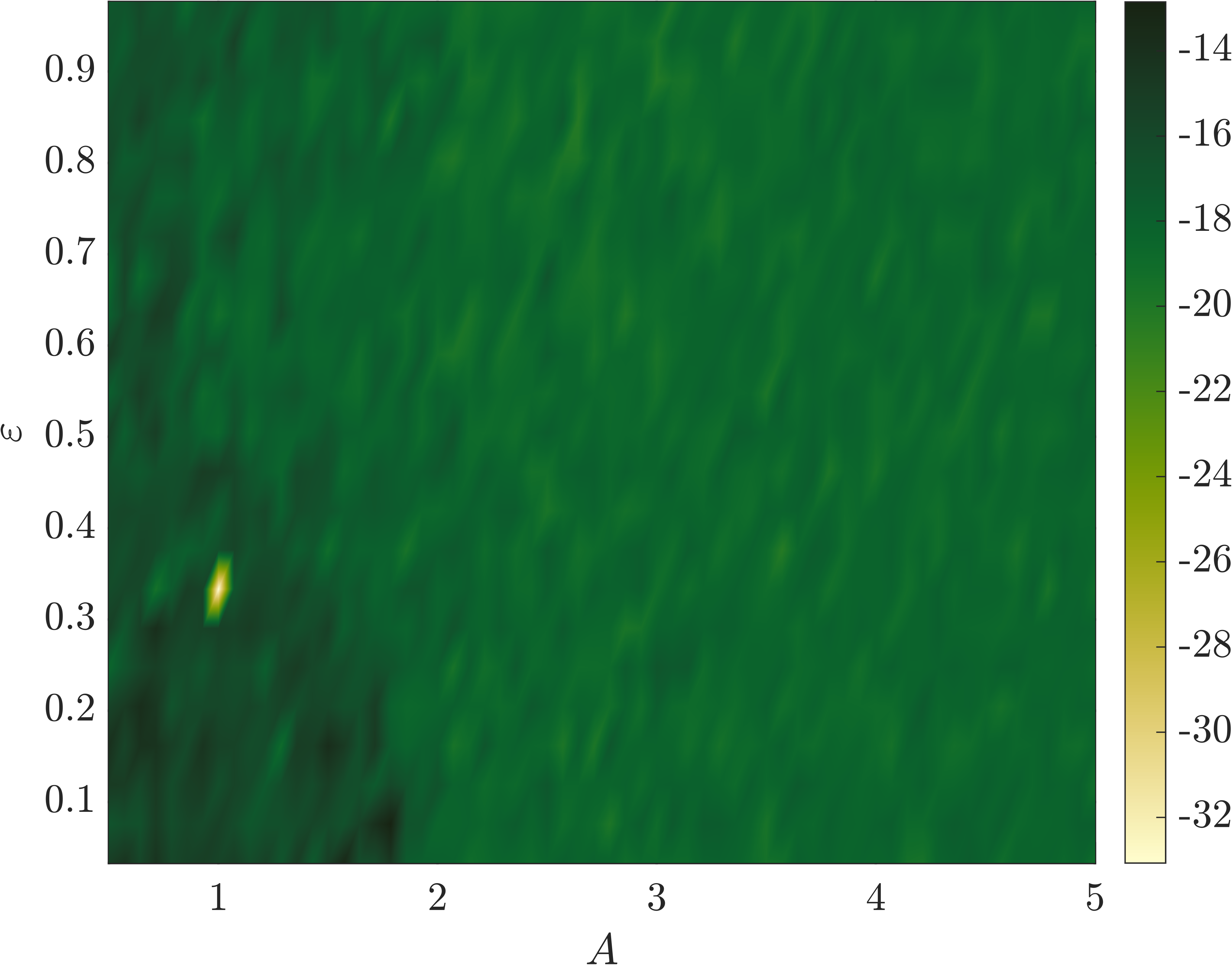}}
\end{centering}
\caption{A broad parameter search (with $B$ fixed to 1) for integrability detection in the HH system given by Hamiltonian~\eqref{eq:HHHam}. The optimization of the loss function, shown on a logarithmic scale, identifies a distinct position at $(A, \varepsilon) = (1, 1/3)$ where integrability is meaningfully detected, differing by several orders of magnitude from background loss values. Cubic spline interpolation of the landscape is used for visualization.}\label{fig:HHdetection}
\end{figure}

\section{The Korteweg-deVries Equation}\label{section:KdV}
We now extend our methodology to Hamiltonian partial differential equations beginning with the Korteweg-deVries (KdV) equation
\begin{equation}\label{eq:KdV}
	\partial_tu-6 u\partial_xu+\partial_x^3u=0.
\end{equation}
We first review the Hamiltonian form of the KdV equation~\cite{zakharov1971korteweg} which is given by
$$
u_t=Q \frac{\delta H}{\delta u}
$$
where
\begin{equation}\label{eq:KdVHam}
	H =\int_{-\infty}^{\infty}\left(u^3-\frac{1}{2} u u_{x x}\right) d x:=\int_{-\infty}^{\infty} h\left(u, u_{x x}\right) d x.
\end{equation}

The KdV equation is bi-Hamiltonian; Two choices of $Q$ 
(with corresponding choices of $H$) 
yield the equation of motion. For the choice of
$H$ used here, the relevant operator is $Q=\partial_x$. Specifically, since the Fr{\'e}chet derivative of the functional $H$ is given by
$$
\frac{\delta H}{\delta h}=\frac{\partial h}{\partial u}-\partial_x \frac{\partial h}{\partial u_x}+\partial_x^2 \frac{\partial h}{\partial u_{x x}}.
$$
we see that
$$
\partial_x \frac{\delta H}{\delta u}=6 u\partial_xu-\partial_x^3u=\partial_tu,
$$
reproduces the KdV equation~\eqref{eq:KdV}.

The KdV equation on the real line constitutes a prototypical example of a Hamiltonian PDE that is integrable. A Lax pair is known~\cite{gardner,ablowitz1981solitons}:
\begin{equation}\label{eq:KdVLax}
\begin{aligned}
	L&=-\partial_x^2+u, \\
	P&=4 \partial_x^3-6 u \partial_x-3 u_x.
\end{aligned}
\end{equation}
The KdV equation can thus be viewed as the compatibility between these differential operators; that is, the equation
$$
\partial_t L=[L, P]
$$
reproduces Equation~\eqref{eq:KdV}.

To yet again build a loss function that does not involve the explicit use of time derivatives, we use the chain rule. Observe that
$$
\begin{aligned}
	\partial_t L & =\frac{\partial L}{\partial u} \frac{\partial u}{\partial t} \\
	&= \frac{\partial L}{\partial u} Q\frac{\delta H}{\delta u}=[L, P].
\end{aligned}
$$
Thus, the expression $\frac{\partial \tilde{L}}{\partial u} Q\frac{\delta H}{\delta u}$ plays the role of the Poisson bracket in the design of the loss functions from previous sections. By direct analogy with Problem~\eqref{eq:SHOProb}, our optimization problem for this Hamiltonian setting is given by
\begin{equation}\label{eq:KdVProb}
	\min _{\eta\in\mathbb{R}^{N_{\eta}}}J[\eta]=\min _{\eta\in\mathbb{R}^{N_{\eta}}}
    \frac{1-r}{N_{\rm samples}}\sum_{j=1}^{N_{\rm samples}} \left.\frac{\int\left|\frac{\partial \tilde{L}}{\partial u} Q \frac{\delta H}{\delta u}u_j-[\tilde{L}, \tilde{P}]u_j\right|^2dx}{\int\left|\frac{\partial \tilde{L}}{\partial u} Q \frac{\delta H}{\delta u}u_j\right|^2dx}\right|_{u_j\in\Omega}+r\mathcal{R}^*(\eta)
\end{equation}
where $\Omega$ is a subset of the KdV phase space and $\mathcal{R}^*$ is the sparsification function given by Equation~\eqref{eq:bestR} earlier.

In this PDE setting, there are only two additional modifications we need to make to the numerical framework from previous sections. First, we must carefully consider how we sample the phase space and how we construct the operator hypotheses. This amounts to sampling from the function space $L^p(\mathbb{R})$, $p>1$. To this end, we construct random samples from the overcomplete basis
\begin{equation}\label{eq:PDEsample}
u^{\textrm rand}(x)=N\sum_j e^{-a_j(x-b_j)^2}\sum_k\frac{A_{jk}}{k^3}\sin\frac{k\pi x}{L}
\end{equation}
where all parameters $a_j,\ b_j, A_{jk}$ are appropriately sampled from uniform distributions, $L$ is the length of the truncated spatial domain, and the coefficient $N$ ensures a unit norm in the space $L^1(\mathbb{R})$. In this way, we try to verify the compatibility of our operators with functions sampled from the function space $C^{3}(\mathbb{R})\cap L^p(\mathbb{R})$, $p>1$ and with equivalent masses in the sense of $L^1(\mathbb{R})$. The smoothness of our samples is ensured by the decay of the Fourier coefficients $A_{j,k}/k^3$ while also being regularized by the fact that we only take the sum over $k$ to be finite. We used $k=10$ throughout our study. 
This smoothness is used for computational tractability in the evaluation of differential operators in the loss. We will comment on how this choice of sampling, as opposed to sampling from any other function space, affects numerical results shortly.

The second modification is the operator hypothesis. In general, Lax pairs are linear differential operators (in $x$) with coefficients dependent on $x,\ u,$ and derivatives of $u.$ A fairly wide class of operators is given by
\begin{align*}
	\tilde{L}&=\xi_1u+\xi_2\partial_x+\xi_3\partial_x^2,\\ \tilde{P}&=\sum_{l} \sum_j \sum_k \sum_m \zeta_{l j k m} x^{l-1} u^{j-1} \left(\partial_x^{k-1} u\right) \partial_x^{m-1}
\end{align*}
where the indices in the sum all start at one. It may seem like the hypothesis on $\tilde{L}$ is overly restrictive. However, we must keep in mind that we seek to reproduce an evolution equation involving $\partial_t\tilde{L}$. Should any higher powers of $u$ enter into the hypothesis of $\tilde{L},$ then the likelihood of finding an \textit{explicit} equation of the form $u_t=F(u,u_x,u_{xx},...)$ decreases substantially.  In this sense, we seek to preserve the
semi-linear (in time) nature of the PDE of interest.

We also comment that the dependence on polynomial coefficients $x$ is redundant when it comes to Lax pairs for the KdV equation. Therefore, to keep our methodology simple, we drop the dependence on the monomials $x^l$ and its associated index on the tensor coefficients $\xi$ and $\zeta.$ We also assume that $P$ has at most third-order coefficients. For clarity on these choices, reproducing the KdV Lax pair amounts to setting $\xi_{1}=1$, $\xi_{3}=-1,$ $\zeta_{4,1,1}=4$, $\zeta_{2,2,1}=-6$, $\zeta_{1,1,2}=-3$ with the remaining 34 coefficients set to zero. Lastly, to compute the spatial derivatives in the operator hypothesis, we use the Fast Fourier transform. That is, we use the formula
$$\partial^j_xu=\mathcal{F}^{-1}\left\{(ik)^j\mathcal{F}\{u\}\right\}$$
to compute the derivatives spectrally, where $k$ denotes the grid-dependent wavenumbers.

Since we are computing derivatives assuming the use of periodic boundary conditions, we sample from $\Omega$, ensuring that all samples have compact support to double precision in the interval $[-20\pi,20\pi]$. Despite using just $2^{11}$ points, we verify that the loss in Problem~\eqref{eq:KdVProb}, evaluated at $r=0$ and at the values of $\eta$ that give the known KdV Lax pair, is zero to double precision over 100 such samples from $\Omega$. 

In our training, we find that only using $20$ samples is sufficient in our cross-validation. After training on these samples, we find that the loss, on average and without sparsification, is on the order of $10^{-11}$ when evaluated on unseen samples from $\Omega.$ We show, in Figure~\ref{fig:KdVResult}, a visual comparison between $\frac{\partial \tilde{L}}{\partial u} Q \frac{\delta H}{\delta u}$ and $[\tilde{L}, \tilde{P}]$ for four unseen sample functions. Note that despite the unit norms of the samples in $L^1(\mathbb{R})$, the generalized Poisson bracket and matrix commutators have arbitrary norms, which account for the variance of the scales across these images.

\begin{figure}[htbp]
	\begin{centering}		\subfigure{\includegraphics[width=0.45\textwidth]{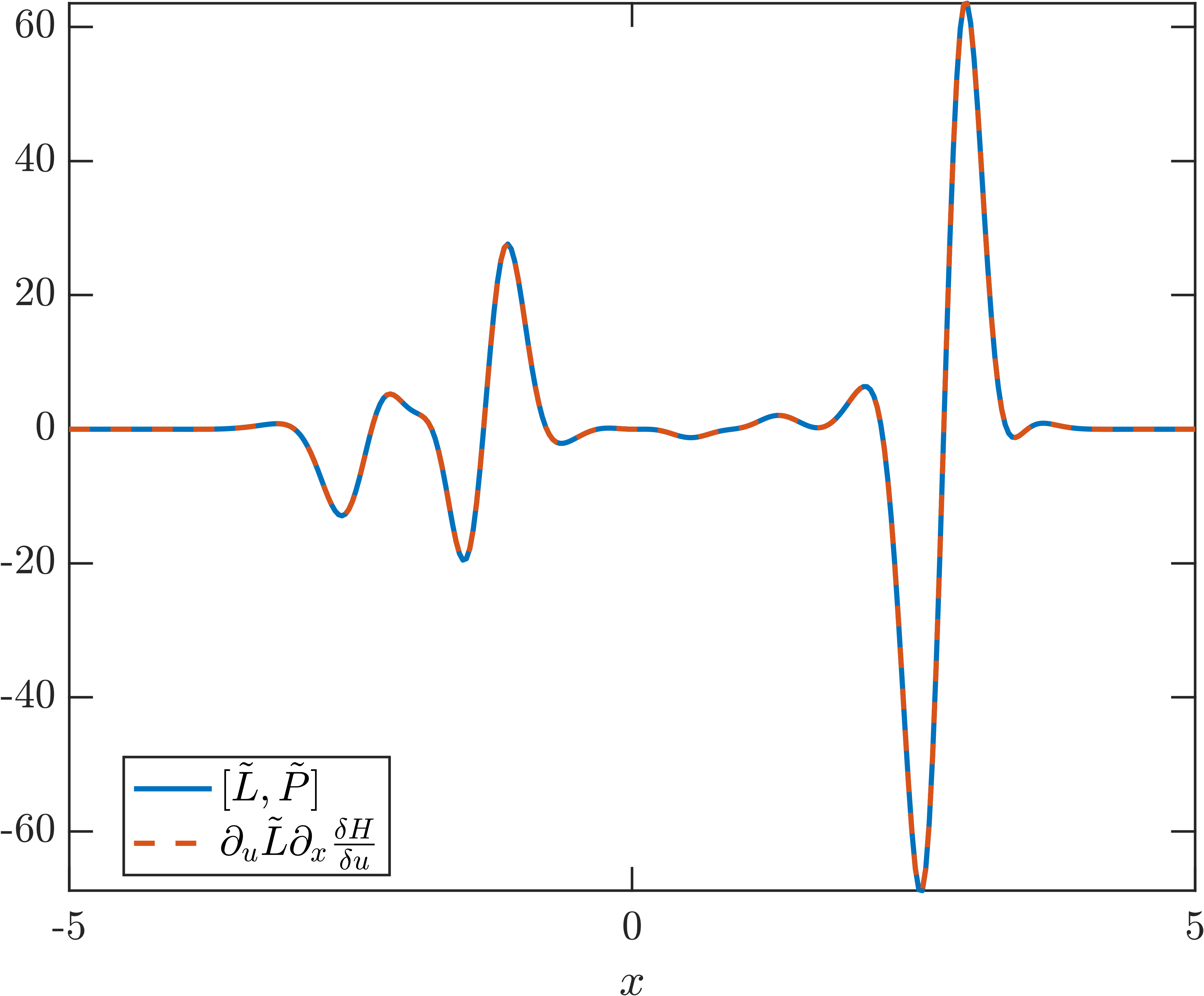}}\subfigure{\includegraphics[width=0.45\textwidth]{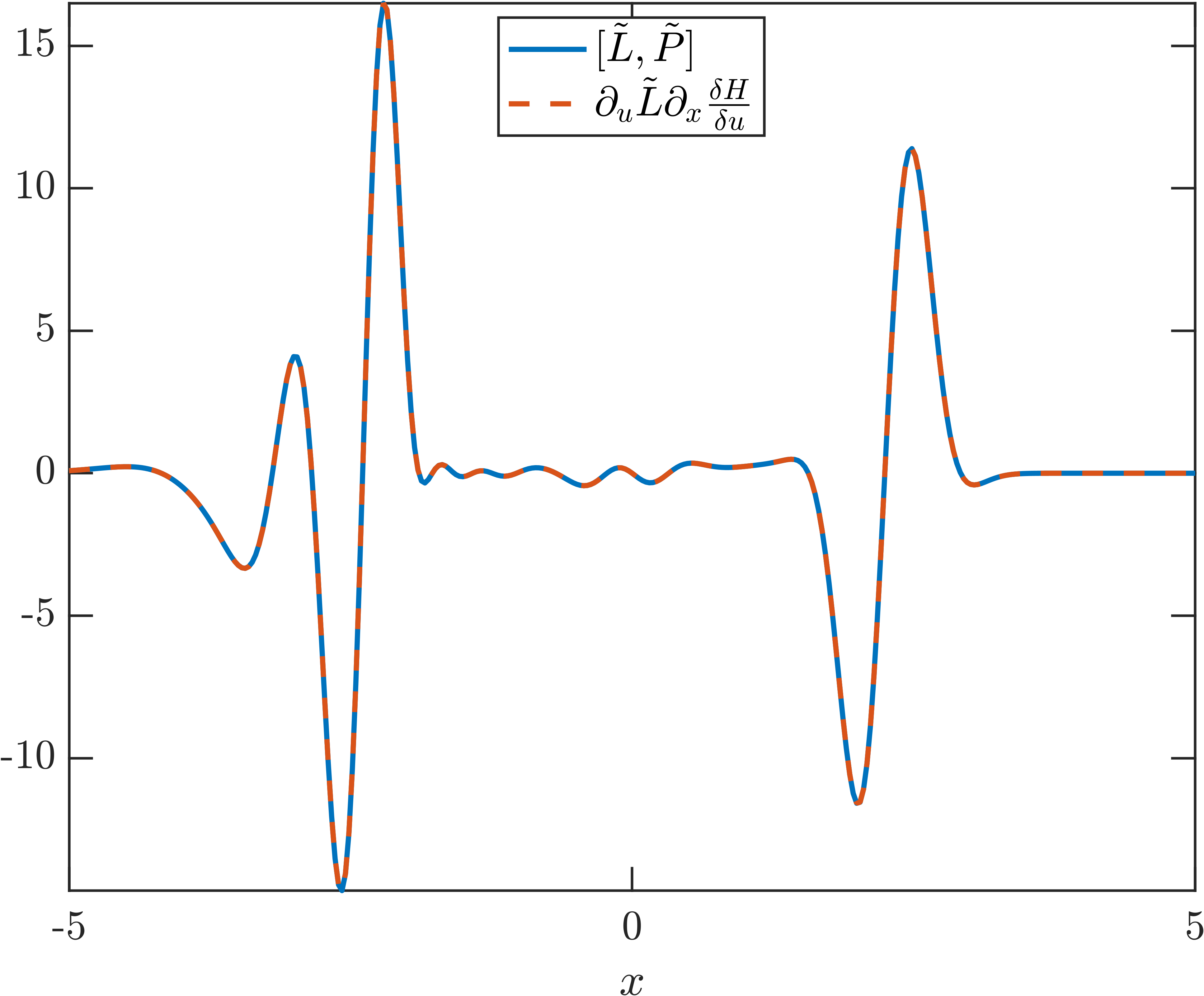}}
\subfigure{\includegraphics[width=0.46\textwidth]{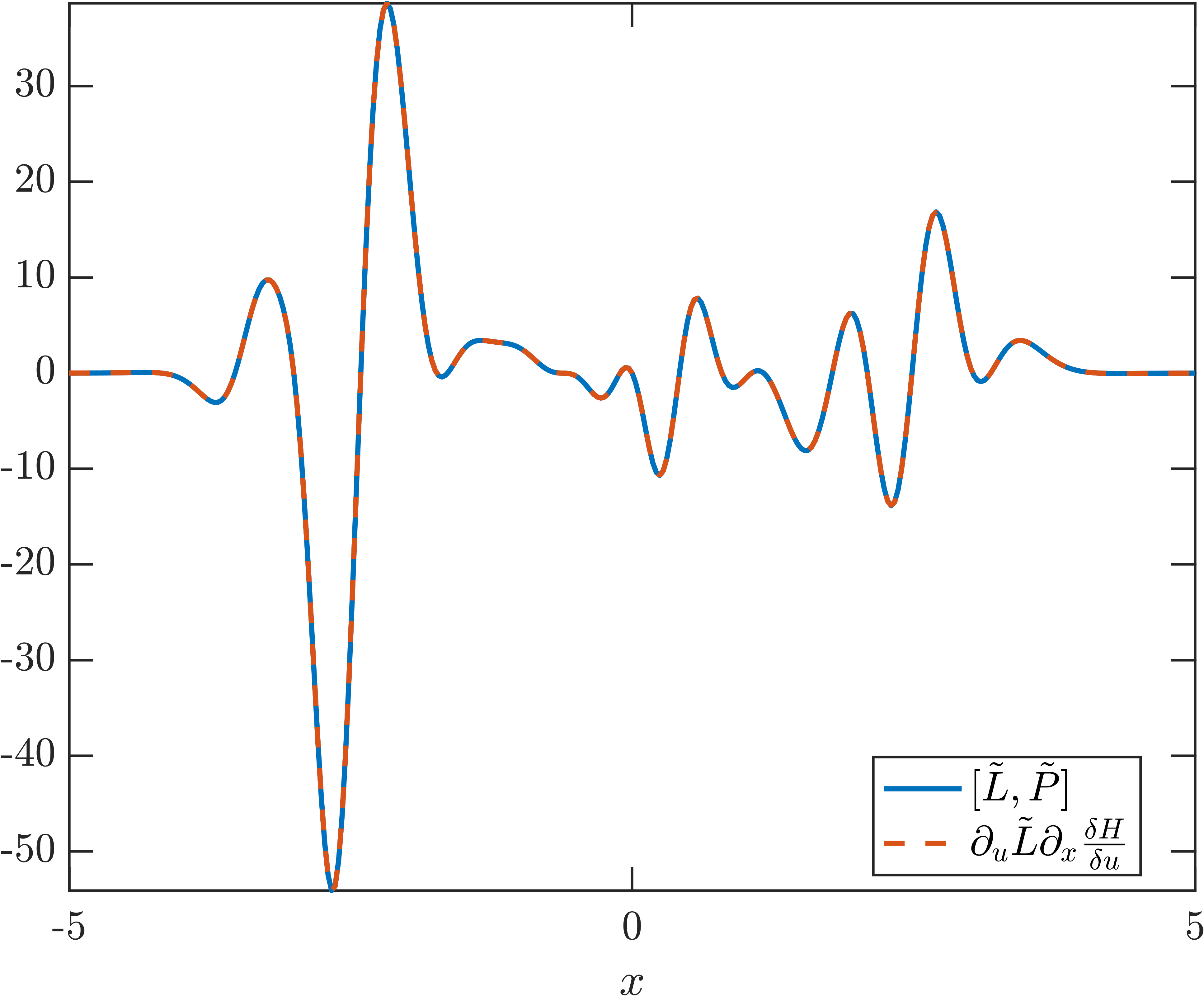}}
\subfigure{\includegraphics[width=0.45\textwidth]{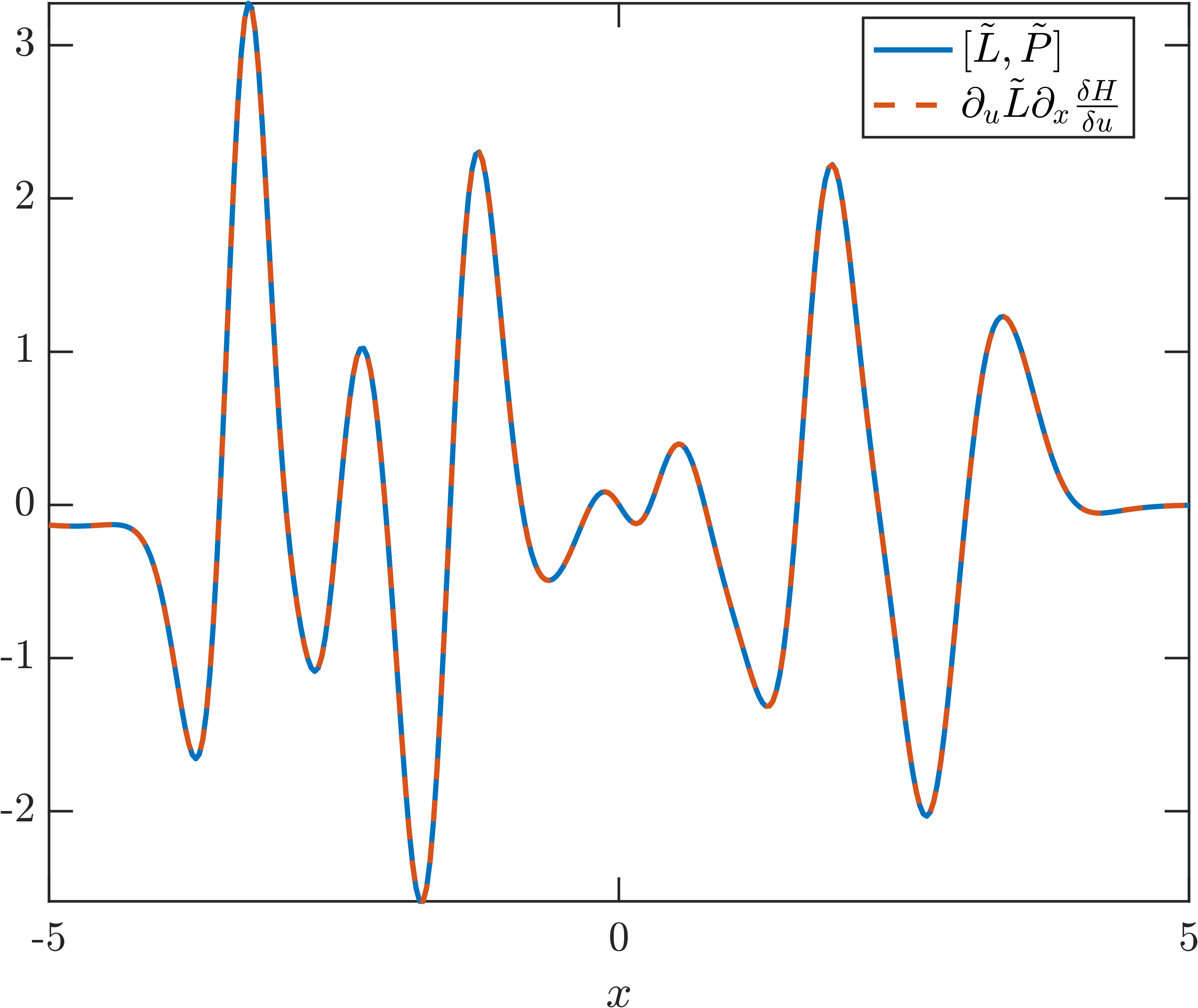}}
	\end{centering}
	\caption{A numerical result of solving Problem~\eqref{eq:KdVProb} without sparsification. Visualized here is a cross-validation study displaying the generalized Poisson brackets and commutators evaluated at the optimal point $\eta^*$ and on four samples from the function space $\Omega$ that were unseen during training. For all four cases, the loss is on the order of $10^{-11}$.}\label{fig:KdVResult}
\end{figure}

The following remark is critical both mathematically and computationally. The careful reader will notice that we have not made a distinction between the functions $u(x)$ used in our operator hypotheses and the function $u(x)$ on which the operators act. That is because, computationally, we do not. This is {\bf a mathematically restricted} interpretation of how Lax pairs work
since, in principle, the relevant action of the operators in the
Lax equation is applicable to {\it arbitrary} functions, that is,
this is valid as an operator equation. 
Nevertheless, despite its limitations, we will show that this \textbf{computational interpretation} of the problem still achieves many of the desired objectives of detecting integrability and discovering sparse Lax pairs.

A natural question arises: why not adopt the interpretation
$$
\left(\partial_tL(u)-[L(u),P(u)]\right)w=0, \quad \forall w\in\mathcal{F}
$$
where $\mathcal{F}$ is the relevant function space? We did pursue this. We used the simple scenario of fixing $u$ as a single function, specifically $u(x)=e^{-x^2}.$ Then we randomly sample 20 functions $w$ using our sampling scheme. We report that our optimization methodology employed up until now absolutely fails to find any Lax pair that minimizes the unbiased loss past values smaller than $10^{-4}.$ This, as we found, turns out not to be an acceptable order of magnitude to resolve the detection of integrability against smooth nonintegrable perturbations; see Figure~\ref{fig:KdVPerturb} to observe the degree of precision required to detect integrability. 
Of course, this is a point that merits further investigation both
mathematically and through computational experiments, especially since this order of magnitude for the loss is often deemed acceptable when using neural networks; see, for example,
~\cite{Liu_2022} (Table II)  or~\cite{zhu2023machine}.

It is rather remarkable that this mathematically restricted interpretation of the optimization problem still yields mathematically correct results. We leveraged substantial computational resources to solve the unrestricted mathematical interpretation of the optimization problem to no avail (at least for the prescribed numerical precision). For this reason, we believe that our interpretation, relative to our current computational framework, is crucial to ensuring the desired degree of tractability for this problem.

To demonstrate the validity of our approach, we show that our framework is sensitive enough to detect the known integrability of the KdV equation. Consider the non-integrable perturbation $h_1=\frac{1}{2}\left(\partial_x^2 u\right)^2$. We solve Problem~\eqref{eq:KdVProb}, again without sparsification, for Hamiltonian densities $h+\varepsilon h_1$, where $h$ is defined in Equation~\eqref{eq:KdVHam} and $\varepsilon\in[-.01,.01]$. Figure~\ref{fig:KdVPerturb} shows that the loss has a nearly smooth dependence on the parameter $\varepsilon$ with a minimum at the expected integrable point $\varepsilon=0$. Again, the integrability point is privileged, as the loss is several orders of magnitude smaller than the nearby values of $\varepsilon.$ 

\begin{figure}[htbp]
	\begin{centering}		\subfigure{\includegraphics[width=.7\textwidth]{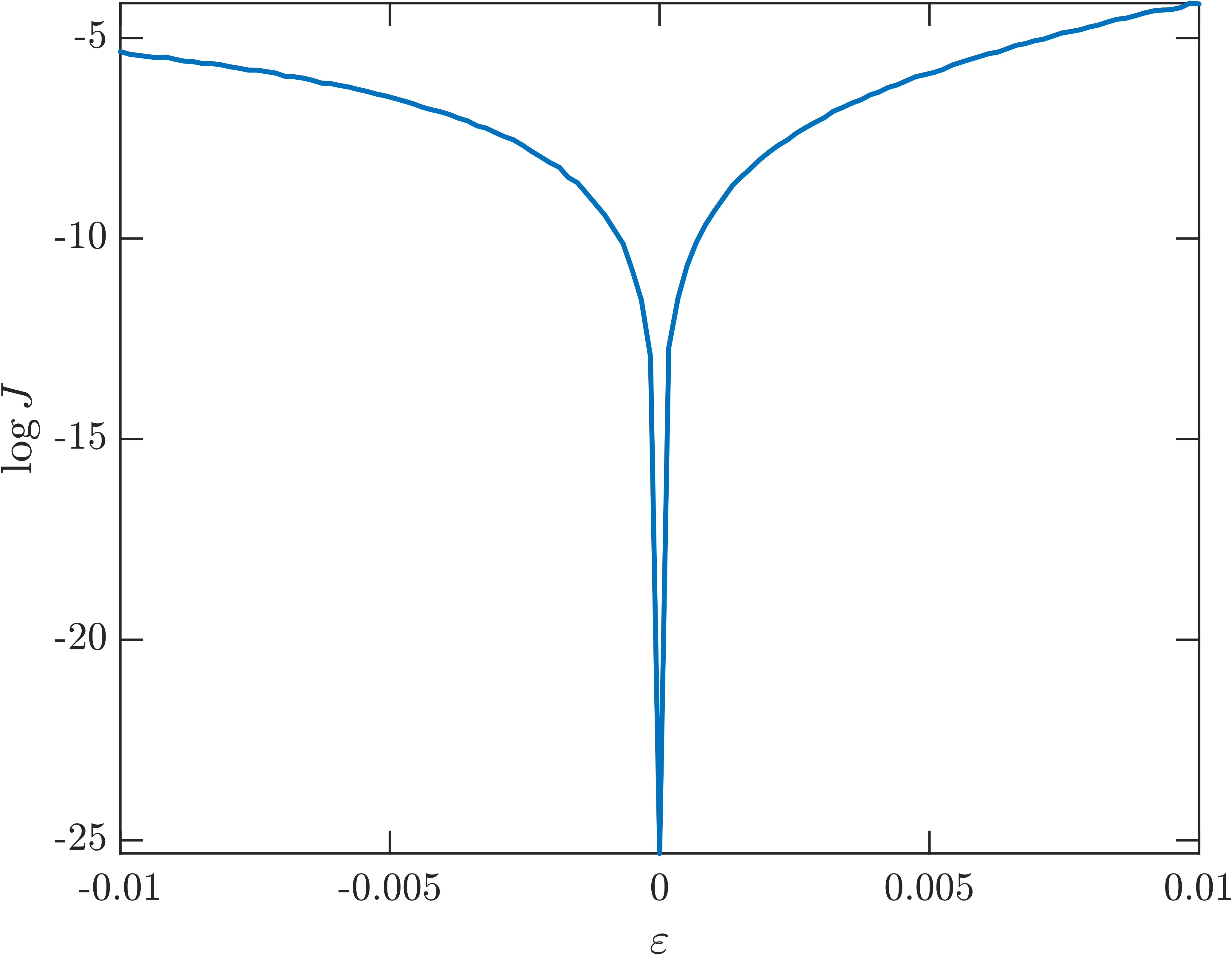}}
	\end{centering}
	\caption{A perturbation study using the perturbed Hamiltonian density $ h_1=\frac{\varepsilon_1}{2}\left(\partial_x^2 u\right)^2$. Shown here is the numerical solution of Problem~\eqref{eq:KdVProb}, without sparsification, with density $h+\varepsilon h_1$. We observe a near-smooth dependence on $\varepsilon$ with a clearly discernible
    ``special'' point associated with the detection of integrability at $\varepsilon=0$. }\label{fig:KdVPerturb}
\end{figure}

We now discuss the interpretability of building Lax pairs from solving Problem~\eqref{eq:KdVProb}. Without sparsification, it is not surprising that all 39 coefficients in our operators are activated. Therefore, even with computer algebra systems such as Mathematica, we have no chance to interpret the PDE that the Lax pair is producing.  Before discussing what we discover through sparsification, we make the following basic observations. 

Suppose that we were fortunate enough to choose the perfect operator hypothesis
$$\begin{aligned} 
	&L=a \partial_x^2+b u,\\
	&P=c \partial_x^3+d u \partial_x+e \partial_xu.
\end{aligned}
$$
We can determine these coefficients by hand and show that they match the operators in Equation~\eqref{eq:KdVLax}, up to a constant factor in the operator $L$. This implies that if our numerical method successfully identifies this Lax pair, it should do so uniquely, modulo a scaling factor. Because of our restricted approach of taking the coefficients in the Lax pair to be the same functions we evaluate on, this introduces some indeterminacy. A calculation by hand shows that the compatibility of these operators, with the interpretation $\partial_tL(u)u=[L(u),P(u)]u,$ yields
$$\begin{aligned}
	a e-c b +b& =0, \\
	-b d-6b&=0, \\
	3 a d+2 a e  -6 c b&=0.
\end{aligned}
$$
This system is, of course, underdetermined featuring $3$ equations
for $5$ unknowns. If we solve Problem~\eqref{eq:KdVProb} with this perfect hypothesis, we find that the combination of parameters identified is such that these equations are satisfied to $\mathcal{O}(10^{-6}).$ Meanwhile, the loss is $\mathcal{O}(10^{-12}).$

We also make the observation that, in principle, we have the freedom to  rescale by space and amplitude so that we can remove two parameters, that is,
$$\begin{aligned} 
	&L=\partial_x^2-u,\\
	&P=c \partial_x^3+d u \partial_x+e u_x
\end{aligned}
$$
and where the three remaining parameters have been appropriately rescaled. Upon solving the optimization problem, we practically recover the known coefficients. We find that $c=4,\ d=6,$ and $e=-3$ to 6 digits of precision with a loss that evaluates to $\mathcal{O}(10^{-13})$. Thus, with all of this fortuitous hindsight, we can rediscover the expected KdV Lax pair with a small loss and with no ambiguity in the coefficients. These observations give us a target and a level of precision to aim for as we employ our sparsification strategy.

Once again, we use the same sparsification from the finite-degree-of-freedom setting to aid us in our interpretation of our discovered Lax pairs. Interestingly, our numerics have a lot of difficulty rediscovering the known KdV Lax pair just discussed and commonly reported in the literature because, in fact, \textit{this is not the most parsimonious Lax pair that exists}. 
Our numerics discover, to the best of our knowledge, two entirely new families of Lax pairs, each containing more parsimonious four-term Lax pairs---instead of the five-term commonly reported one in Equation~\eqref{eq:KdVLax}---as special cases. Before discussing these more parsimonious Lax pairs, we state the discovered Lax pair in its full generality.

\begin{theorem}[Existence of a new KdV Lax pair]\label{thm:strongthm}
For every $u\in C^1([0,T];C^3(\mathbb{R}))$, there exists a parameter $v\in\mathbb{R}$ such that the pair of operators
    $$\begin{aligned} 
	&L=\alpha u+\beta\partial_x,\\
	&P=\gamma u+\delta u^2+\epsilon u_{xx}+\kappa\partial_x,
\end{aligned}
$$
satisfying Lax's equation $\partial_tL=[L,P],$ understood as acting on the function space $C^3(\mathbb{R}),$ reproduces the KdV equation
$$
u_t=\frac{2 \beta  \delta}{\alpha }  u u_x+\frac{\beta  \epsilon}{\alpha }u_{\text{xxx}}
$$
in the co-traveling reference frame $x\to x-vt$.
\end{theorem}
\begin{proof}
    By direct calculation, observe that for every $w\in C^3(\mathbb{R}),$ we have
    $$
    [L,P]w=LPw-PLw=\left((\beta  \gamma -\alpha  \kappa)  u_x+2 \beta  \delta  u u_x+\beta  \epsilon  u_{xxx}\right)w.
    $$
    Since $\partial_tL=\alpha u_t$ we see that the operator equation $\partial_tL=[L,P]$ reproduces the desired KdV equation after the appropriate Galilean boost of velocity $v=\frac{\beta  \gamma -\alpha  \kappa}{\alpha}$.
\end{proof}

Before moving forward with interpreting this result, we comment that we are aware that the function spaces for which we state our results are classical. 
In fact, it is well known that the Cauchy problem for the KdV equation is well-posed for initial conditions in the Sobolev space $H^s(\mathbb{R})$, $s>3/4$~\cite{kenig1991well} and even more recently it was discovered that, incredibly, the problem is well-posed in $H^{-1}(\mathbb{R})$~\cite{killip2019kdv}. For these reasons, it is obvious that we could employ standard embedding techniques~\cite{evans2022partial} to relax the regularity assumptions that we make. We find this unnecessary at this stage, as this would detract from the essence of this finding. Therefore, we proceed with classical regularity assumptions that are compatible with the discovered Lax pairs.

In light of Theorem~\ref{thm:strongthm}, the most parsimonious Lax pair has four, not five, terms. Specifically, this is the valid choice of $\gamma=\kappa=0.$ In this degenerate case, the operator $P$ is not a differential operator, as it is purely multiplicative. This degenerate example is consistently found by SILO. We delay a discussion of the modification made to find the Lax pair of Theorem~\ref{thm:strongthm} in full generality until the end of this section. This is because the second type of Lax pair that SILO consistently finds is one that only makes sense under an integral with functions that vanish on the boundary, which is what is used in the evaluation of the loss function. 
The following theorem expresses our discovery of these \textit{weak} Lax pairs precisely.

\begin{theorem}[Existence of a Weak Lax Pair]\label{thm:weakthm} Denote the Schwartz space of distributions by $\mathcal{S}(\mathbb{R})$. Then, for every $u\in H^2([0,T];H^4(\mathbb{R}))$, the Lax pair
    $$\begin{aligned} 
	&L=\alpha u+\beta\partial_x^2,\\
	&P=\delta u_x+\varepsilon u\partial_x,
\end{aligned}
$$
satisfying 
$$\int_{\mathbb{R}}\left( \partial_tL-[L,P]\right)\varphi(x) dx=0, \quad\forall \varphi\in\mathcal{S}(\mathbb{R})
$$
reproduces the equation
\begin{align}
\label{eq:weak_kdv}
    \alpha u_t+\beta(\varepsilon-\delta)u_{xxx}+\alpha\varepsilon uu_x=0
\end{align}
Lebesgue almost everywhere.
\end{theorem}

\begin{proof}
    By direct calculation, we compute the left hand side of the equation $\langle (\partial_tL-[L,P])\varphi,1\rangle_{L^2(\mathbb{R})}=0,$ in the sense of distributions, and find that
    \begin{align}
    \label{eq:weak_lax_pair_proof}
         \langle \alpha u_t-\beta\delta u_{xxx}+\alpha \epsilon uu_x,\varphi\rangle_{L^2(\mathbb{R})}-\beta(2\delta+\varepsilon)\langle u_{xx}  ,\varphi_x\rangle_{L^2(\mathbb{R})}-2\beta\epsilon\langle  u_x,\varphi_{xx}\rangle_{L^2(\mathbb{R})}=0.
    \end{align}
    Integrating by parts once in the second inner product and integrating by parts twice in the third reproduces the weak form of Equation~\eqref{eq:weak_kdv}. 
    Inferring the strong form of the KdV equation follows directly from the continuous embeddings of $H^2([0,T])\hookrightarrow C^1([0,T])$ and $H^4(\mathbb{R})\hookrightarrow C^3(\mathbb{R})$~\cite{evans2022partial}.
\end{proof}

We note that despite how the Lax pair just discussed in Theorem~\ref{thm:weakthm} reproduces the KdV equation in the sense of distributions, one can consider altering the function space so that the Lax pair is a strong one. 
 Suppose that instead of the wide function space $\mathcal{S},$ we operate on functions $\varphi\in  C^3(\mathbb{R})$ such that $L\varphi=\lambda\varphi$ for all $u\in C^1([0,T];C^3(\mathbb{R})).$  Additionally, if $2\delta+\varepsilon=0,$ then the strong Lax equation $\partial_t L=[L,P]$, with this Lax pair being the one stated in Theorem~\ref{thm:weakthm}, reproduces the KdV equation
\begin{align*}
    \alpha u_t-\beta \delta u_{x x x}+3 \alpha\varepsilon uu_x-2\varepsilon\lambda u_x=0.
\end{align*}
As was done in the proof of Theorem~\ref{thm:strongthm}, we may once again exploit the Galilean invariance of the KdV equation to frame boost away the parameter $\lambda$ through $x\to x+\frac{2\lambda\varepsilon}{\alpha} t$. Thus, we strengthen the sense in which the Lax pairs of Theorem~\ref{thm:weakthm} produce an unambiguous KdV equation without the need to integrate by parts. Note, however, that this sense of Lax pair is not fully mathematically general, as it operates only on a subspace of $C^3(\mathbb{R})$ and not the whole space of functions for which classical Lax pairs are typically thought to operate on.   Therefore, we do not claim that this Lax pair with this interpretation has the same mathematical relevance as the typical Lax pair given by Equation~\eqref{eq:KdVLax} and leave this as a curious observation in passing.

Shifting back to our interpretation of the numerics, we see that SILO indeed serves as a computationally aided proof technique in discovering these new Lax pairs, since our conclusions are indeed mathematical and only the means of discovery were computational. Although not entirely insightful, we still report the coefficients we discovered during optimization in Figure~\ref{fig:NewKdV}. There, we show coefficients that achieved unbiased losses of $\mathcal{O}(10^{-15})$. The strong Lax pair reproduces the KdV equation consistent with the Hamiltonian~\eqref{eq:KdVHam} to 8 digits of precision. The weak Lax pair, however, does not produce the KdV equation expected by Theorem~\ref{thm:weakthm}. This is because the loss function~\eqref{eq:KdVProb} involves squaring inside the integral, and integration by parts cannot be carried out as cleanly as was done in the proof of Theorem~\ref{thm:weakthm}. A tedious calculation reveals that the numerical compatibility of the coefficients is captured by the equations $\varepsilon=-6$ and $5\alpha+3\beta(\delta-\varepsilon)=0,$ where these symbols are consistent with Theorem~\ref{thm:weakthm}. Indeed, our numerically discovered coefficients satisfy these equations to 8 digits of precision.

Our findings, shown in Figure~\ref{fig:NewKdV}, show that our numerics discover 5 terms. However, each of our new KdV Lax pairs is a family with only four terms.  This is because the coefficients $\eta_4$ shown in both panels correspond to the constant operator appearing in $\tilde{P}$, that is, where all the indices produce $\zeta_{1,1,1}$. The calculations in the proofs of Theorems~\ref{thm:strongthm} and~\ref{thm:weakthm} with the additional constant $\gamma$ added to the operators $\tilde{P}$ would show that the reproduced KdV equations have no terms involving $\gamma.$ Therefore, the numerically discovered Lax pairs actually only have four meaningful coefficients.

\begin{figure}[htbp]
	\begin{centering}		\subfigure{\includegraphics[width=0.45\textwidth]{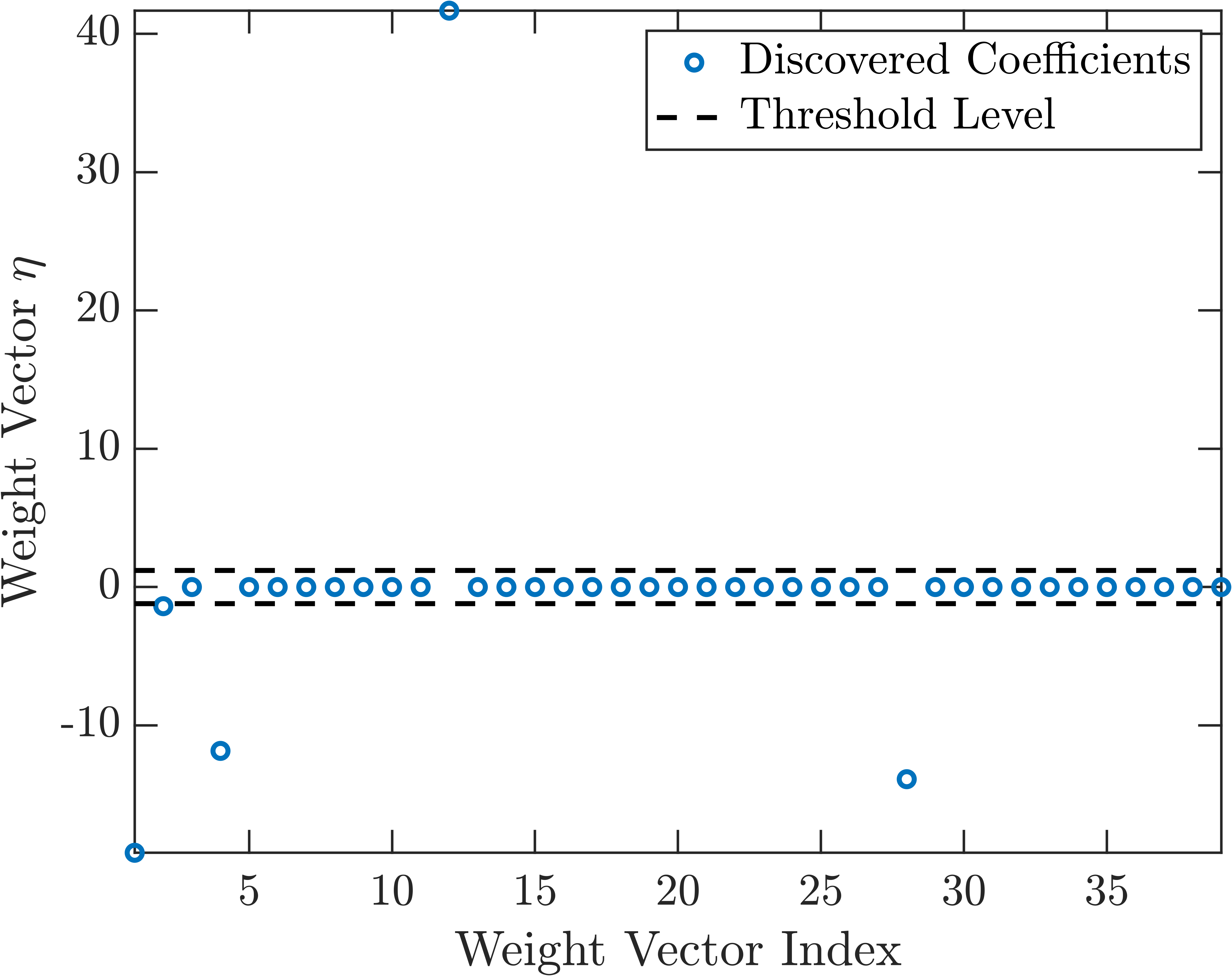}}
    \subfigure{\includegraphics[width=0.45\textwidth]{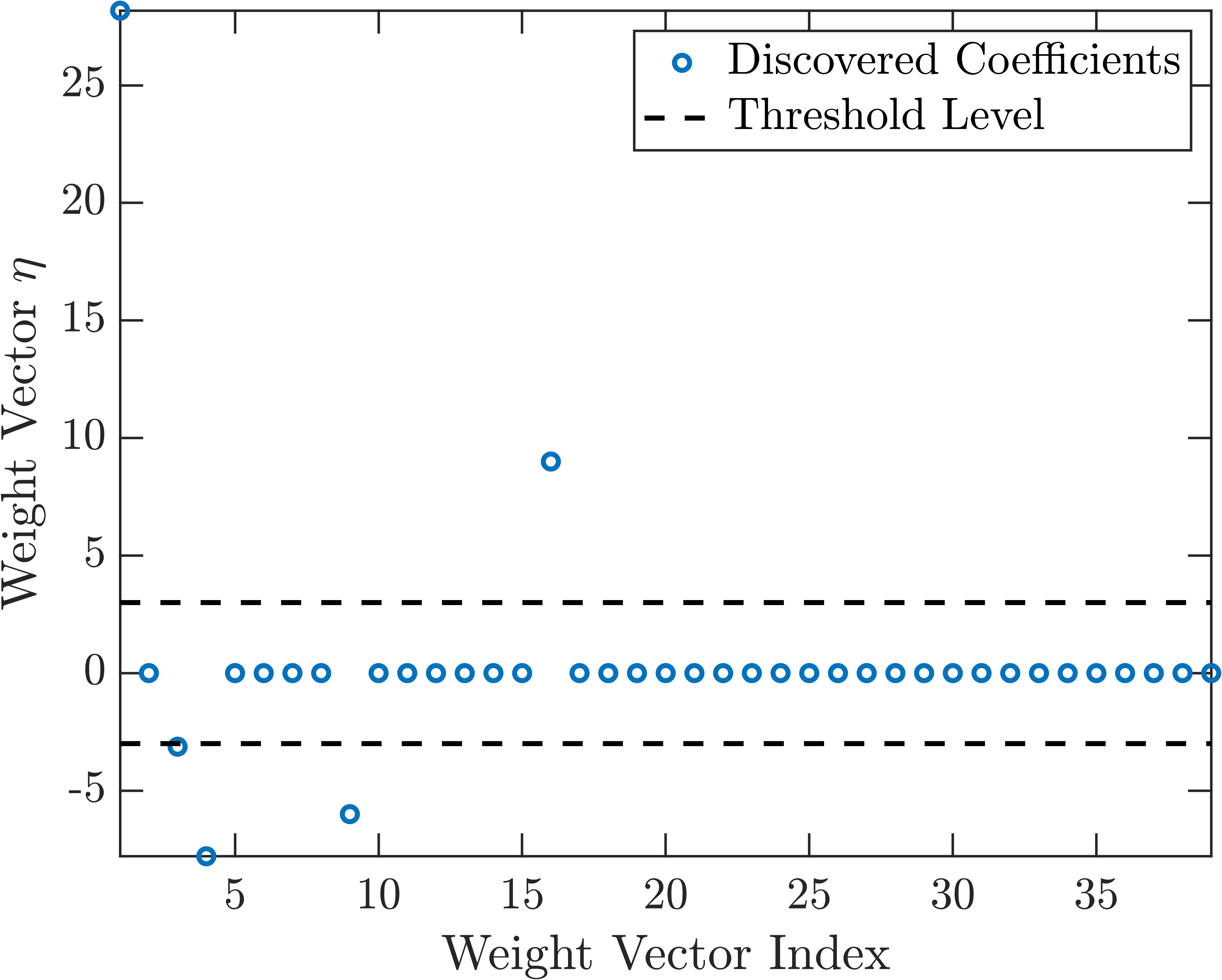}}
	\end{centering}
	\caption{Coefficients discovered during the sparse KdV Lax pair regression as given by Problem~\eqref{eq:KdVProb}. Coefficients shown in the left panel correspond to the degenerate ($\gamma=\kappa=0$) family of Lax pairs discussed in Theorem~\ref{thm:strongthm} while coefficients in the right panel correspond to the family of weak Lax pairs discussed in Theorem~\ref{thm:weakthm}. As discussed in the text, only four coefficients matter since $\eta_4$ in both panels do not contribute in either case to the equation compatible with the computed Lax pairs. In both cases, the loss with $r=0$ evaluates to $\mathcal{O}(10^{-15}).$}\label{fig:NewKdV}
\end{figure}
It is interesting to note that our numerics did not discover the known Lax pair. This is perhaps best explained by the fact that by searching for the sparsest possible Lax pairs, we missed the classical Lax pair with 5 terms.  Indeed, the unbiased loss for the four-term Lax pair is at least two orders of magnitude smaller than the loss when we searched for the best Lax pairs assuming the five-term hypothesis without sparsification.  

For this reason, we searched once again, sampling from the space of periodic functions on $x\in(0,2\pi)$ that do not necessarily vanish at the boundary. We did this to attempt to circumvent the repeated discovery of weak and degenerate Lax pairs. Only then did we discover the fully general six-term Lax pair of Theorem~\ref{thm:strongthm}. Furthermore, we were able to recover the well-known Lax pair when the search was restricted to operators of the form $\tilde{L}=\alpha u+\beta\partial_x^2$.
 Moreover, there exist alternative
formulations of the Lax pair compatibility, such as, 
e.g., that of~\cite{calogerolax}.

Therefore, increasing SILO's thoroughness (or potentially modifying
its setup) poses some intriguing challenges for
further efforts in this direction. Indeed, it does not escape us that there still remain several options for sampling function spaces and evaluating operators, and it is unclear how these may affect discoveries of Lax pairs. We simply report here how our choices explored thus far enable us, as computational users, to discover suitable Lax pairs for the problem of interest. We leave discussions about the mathematical implications of widening the space of test functions to accommodate the existence of weak Lax pairs to Section~\ref{section:conc}.

\section{The Cubic Nonlinear Schr\"odinger Equation}\label{section:NLS}
Our last example involves the focusing nonlinear Schr\"odinger (NLS) equation
\begin{equation}\label{eq:NLS}
i \partial_t \psi=-\partial_x^2 \psi+\frac{2}{p^2-1}|\psi|^2 \psi,
\end{equation}
where $0<p<1$ is a free parameter and $x\in\mathbb{R}$.  In principle, the parameter $p$ can be absorbed by rescaling the equation's variables, yet we maintain its presence so that the discussed Lax pairs are historically aligned with the literature~\cite{zakharov1972exact}.   We make the arbitrary choice of $p=1/\sqrt{2}$  when executing computations.

The NLS is a Hamiltonian system in the sense that  
$i\partial_t\psi=\delta H/\delta \psi^*$ with Hamiltonian
\begin{equation}\label{eq:NLSHam}
H=\int_{-\infty}^{\infty}\left(|\partial_x\psi|^2+\frac{1}{1-p^2}|\psi|^4\right)dx:=\int_{-\infty}^{\infty}h(\psi,\psi^*,\partial_x\psi,\partial_x\psi^*)dx.
\end{equation}
In this way, $\psi^*$ plays the role of the conjugate variable. Using the Hamiltonian, we build the Poisson bracket once more:
$$
\begin{aligned}
	\frac{\partial L}{\partial t} & =\frac{\partial L}{\partial \psi} \frac{\partial \psi}{\partial t}+\frac{\partial L}{\partial \psi^*} \frac{\partial \psi^*}{\partial t} \\
	& =\frac{\partial L}{\partial \psi}\left(-i \frac{\delta H}{\delta \psi^*}\right)+\frac{\partial L}{\partial \psi^*}\left(i \frac{\delta H}{\delta \psi}\right) \\
	& =i\left(\frac{\partial L}{\partial \psi^*} \frac{\delta H}{\delta \psi}-\frac{\partial L}{\partial \psi} \frac{\delta H}{\delta \psi^*}\right) \\
	& =i\{L, H\}=i[L,P]
\end{aligned}
$$
where the Fr{\'e}chet derivatives are given by
$$
\begin{aligned}
	& \frac{\delta H}{\delta \psi}=-\partial_x^2 \psi^*-\frac{2}{1-p^2} \left|\psi\right|^2\psi^*, \\
	& \frac{\delta H}{\delta \psi^*}=-\partial_x^2 \psi+\frac{2}{1-p^2} |\psi|^2\psi.
\end{aligned}
$$

To build the operator hypothesis, 
motivated in part by the well-known Lax pair of the NLS
due to Zakharov and Shabat~\cite{zakharov1972exact},
we assume that the operators in $L$ are at most linear
(in the field $\psi$) and the operators in $P$ are at most quadratic with constant and complex matrix coefficients in $2\times2$. That is, the hypothesis is
$$
\tilde{L}=\sigma_1 \partial_x+\sigma_2 \psi+\sigma_3 \psi^*
$$
$$
\begin{aligned}
	\tilde{P}=\sigma_4 \partial_x^2 & +\sigma_5 \partial_x +\sigma_6 \partial_x \psi+\sigma_ 7\partial_x \psi^*+\sigma_8 |\psi|^2+\sigma_9\psi^2+\sigma_{10}\psi^{*2}+\sigma_{11}\psi+\sigma_{12}\psi^*
\end{aligned}
$$
This amounts to the discovery of 12 matrices $\sigma_j\in\mathbb{C}^{2\times2}$, or 96 real parameters, in the optimization problem
\begin{equation}\label{eq:NLSProb}
	\min_{\eta\in\mathbb{R}^{N_{\eta}}}J[\eta]=\min _{\eta\in\mathbb{R}^{N_{\eta}}} (1-r)\left[ \frac{\int_{\mathbb{R}}\left|\{\tilde{L}, H\}-[\tilde{L}, \tilde{P}]\right|^2dx}{\int_{\mathbb{R}}\left|\{\tilde{L}, H\}\right|^2dx}\right]_\Omega+r\mathcal{R}^*(\eta)
	\end{equation}
where the evaluation over $\Omega$ is understood in the same sense of Equation~\eqref{eq:KdVProb} but over two component functions drawn from the NLS phase space. To sample from this phase space, we similarly use the overcomplete basis given by~\eqref{eq:PDEsample}, replacing the sine function with a complex exponential. We make the computational choice to compute the evaluation of the operators in the sense that they act on the column vector $w(x)=[u(x)\  u(x)^*]^{\intercal}.$ We, once again, emphasize that this is not the most general mathematical choice but has been
empirically found to suffice for our purposes.

We note that, due to Zakharov and Shabat, the Lax pair for the NLS equation is available in the form~\cite{zakharov1972exact}:
\begin{equation}\label{eq:ZakShab}
\begin{aligned}
	& L=i\left(\begin{array}{cc}
		1+p & 0 \\
		0 & 1-p
	\end{array}\right) \partial_x+\left(\begin{array}{cc}
		0 & \psi^* \\
		\psi & 0
	\end{array}\right) \\
	& P=-p\left(\begin{array}{cc}
		1 & 0 \\
		0 & 1
	\end{array}\right) \partial_x^2+\left(\begin{array}{cc}
		\frac{|\psi|^2}{1+p} & i \partial_x\psi^* \\
		-i \partial_x\psi & -\frac{|\psi|^2}{1-p}
	\end{array}\right).
\end{aligned}
\end{equation}
Therefore, we know that our operator hypothesis is wide enough to enclose the Zakharov-Shabat Lax pair. However, despite our best efforts, we were only able to achieve a loss on the order of $10^{-7}$. We believe that with 96 parameters, the loss landscape of this optimization problem becomes highly nonconvex, with numerous local minima preventing convergence to the global solution that coincides with the Zakharov-Shabat solution. To achieve a higher precision detection of integrability, we discard the matrices $\sigma_9-\sigma_{12}$ because, as can be seen, these matrices are not relevant to the resulting equation of motion. Therefore, the optimization problem has a more manageable 64 parameters to find. We defer the exploration of more sophisticated numerical optimization tools for higher dimensional computations to future studies.

With 64 parameters, we are able to minimize the loss to $\mathcal{O}(10^{-10})$. We perform another cross-validation study and visualize the results in Figure~\ref{fig:NLSCV}. Since NLS generalized Poisson brackets and commutators have 4 components, we display the real and imaginary parts of the first row, first column components for different unseen samples in Figure~\ref{fig:NLSCV}. We again see that our numerical optimization generalizes to unseen samples.
\begin{figure}[htbp]
	\begin{centering}		\subfigure{\includegraphics[width=0.45\textwidth]{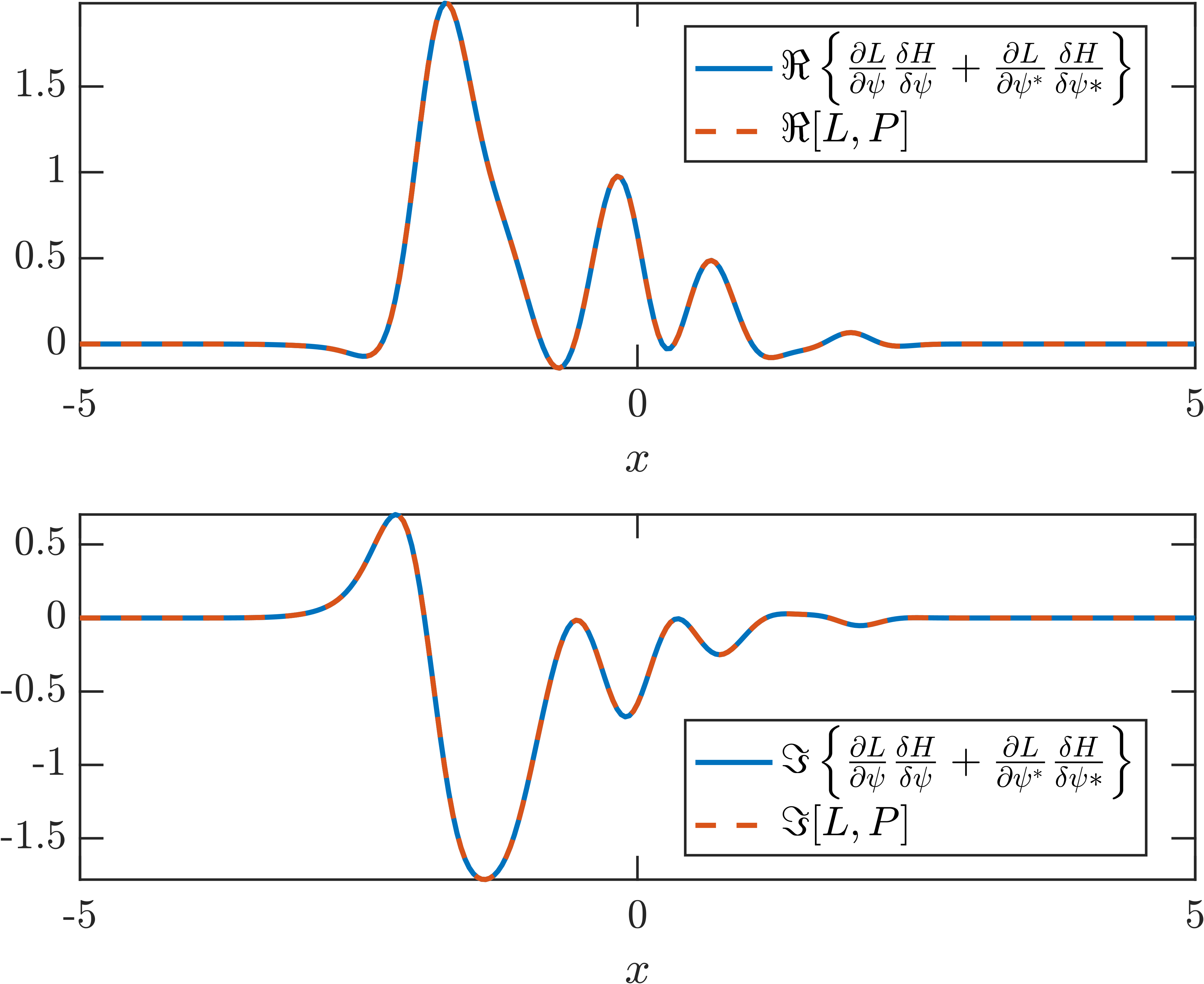}}
\subfigure{\includegraphics[width=0.45\textwidth]{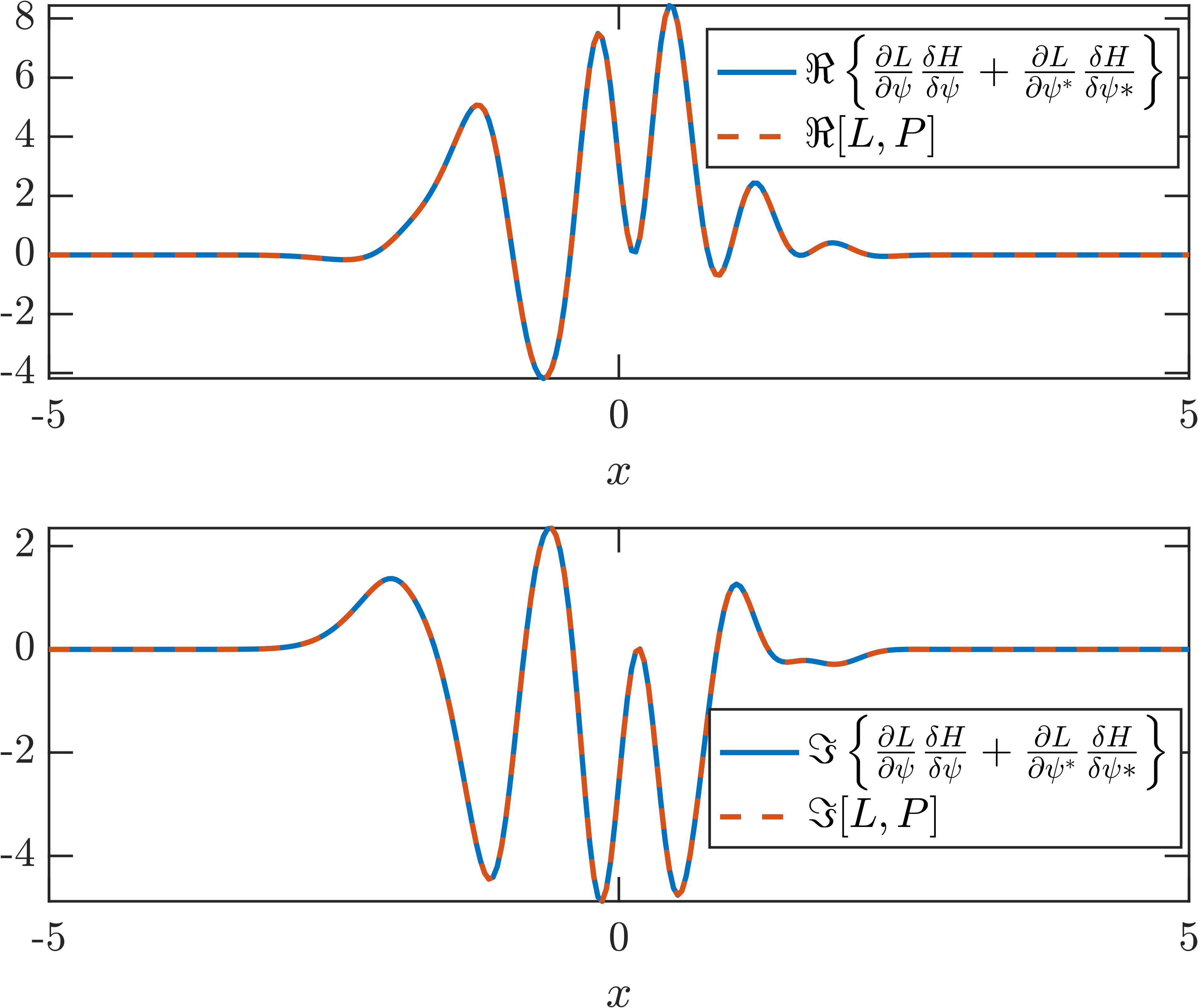}}
\subfigure{\includegraphics[width=0.45\textwidth]{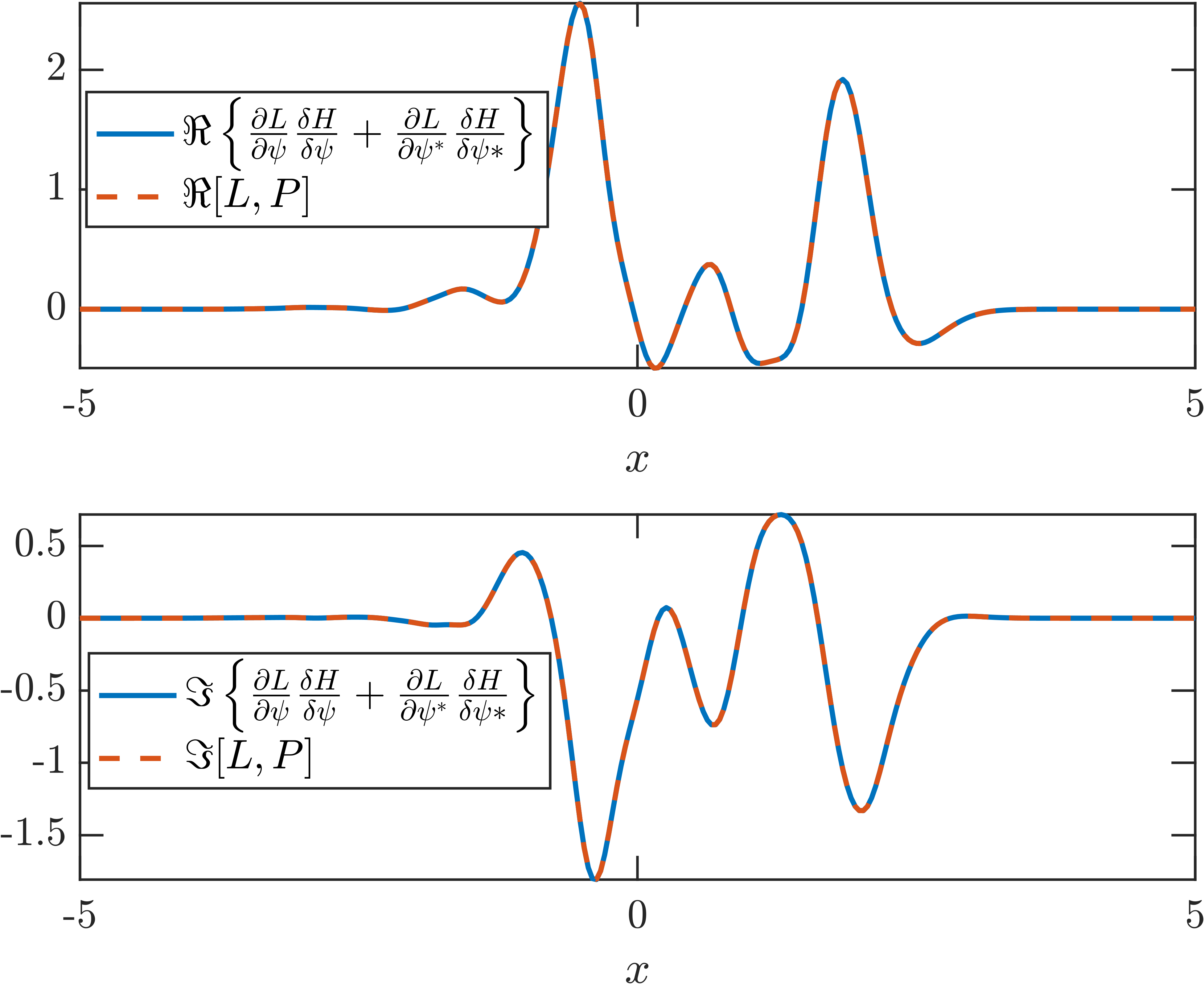}}
\subfigure{\includegraphics[width=0.45\textwidth]{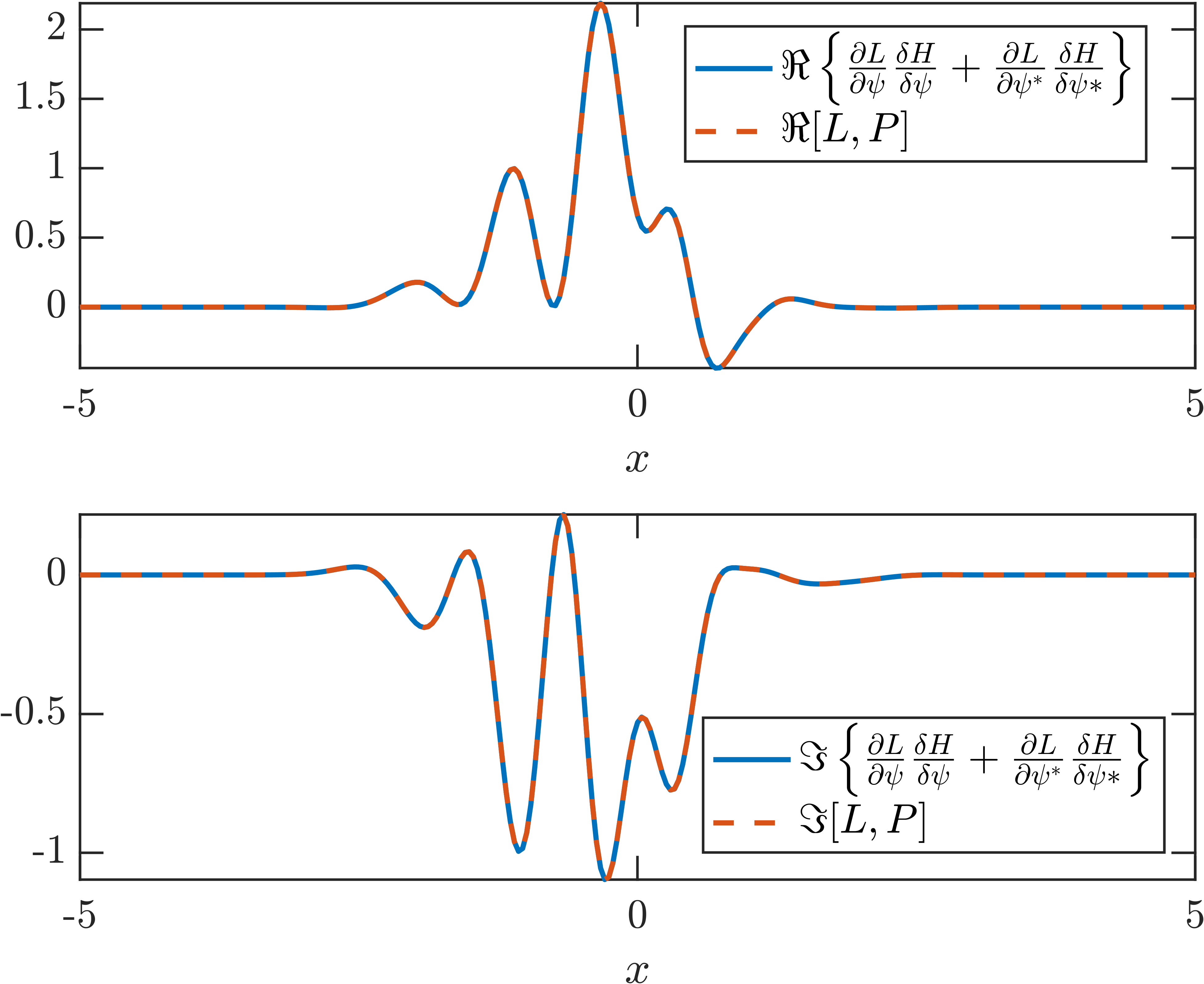}}
	\end{centering}
	\caption{A numerical result of solving Problem~\eqref{eq:NLSProb} without sparsification and with the 64 parameter hypothesis discussed in the text. Visualized here is a cross-validation study displaying the Poisson brackets and commutators evaluated at the optimal point $\eta^*$ and on four samples from the function space $\Omega$ that were unseen during training. For all four cases, the loss is on the order of $10^{-10}$. Displayed are   the real and imaginary parts of the first components of the vectors resulting from the evaluation of the Lax pairs.}\label{fig:NLSCV}
\end{figure}

To investigate again the sensitivity of the integrability detection, we introduce non-integrable Hamiltonians
$h_1=\frac{1}{3}\left|\psi\right|^6$ and $h_2=\frac{1}{2}\left|\partial_x^2 \psi\right|^2$. We solve Problem~\eqref{eq:NLSProb}, without sparsification, for Hamiltonain densities $h+\varepsilon_1 h_1+\varepsilon_2h_2$, where $h$ is defined in Equation~\eqref{eq:NLSHam} and $\varepsilon_1,\varepsilon_2\in[-.01,.01]$. We show in Figure~\ref{fig:NLSPerturb} that the loss has distinguished minima at the integrable points $\varepsilon_1=0$ (for fixed $\varepsilon_2=0$) and
$\varepsilon_2=0$ (for fixed $\varepsilon_1=0$). Once again, we conclude our methodology is precise enough to detect the integrability of the PDE under study.

\begin{figure}[htbp]
	\begin{centering}		\subfigure{\includegraphics[width=0.45\textwidth]{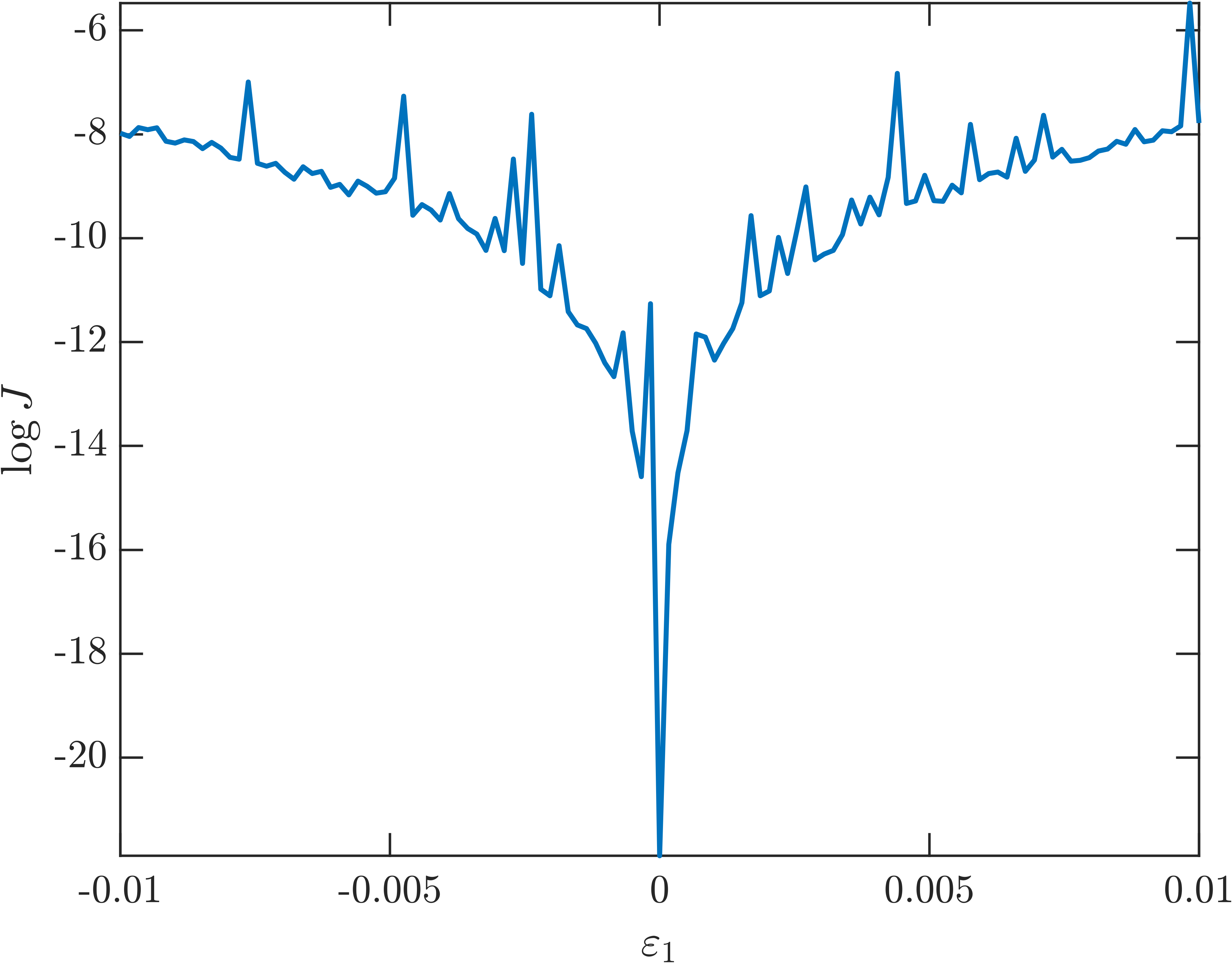}}
\subfigure{\includegraphics[width=0.45\textwidth]{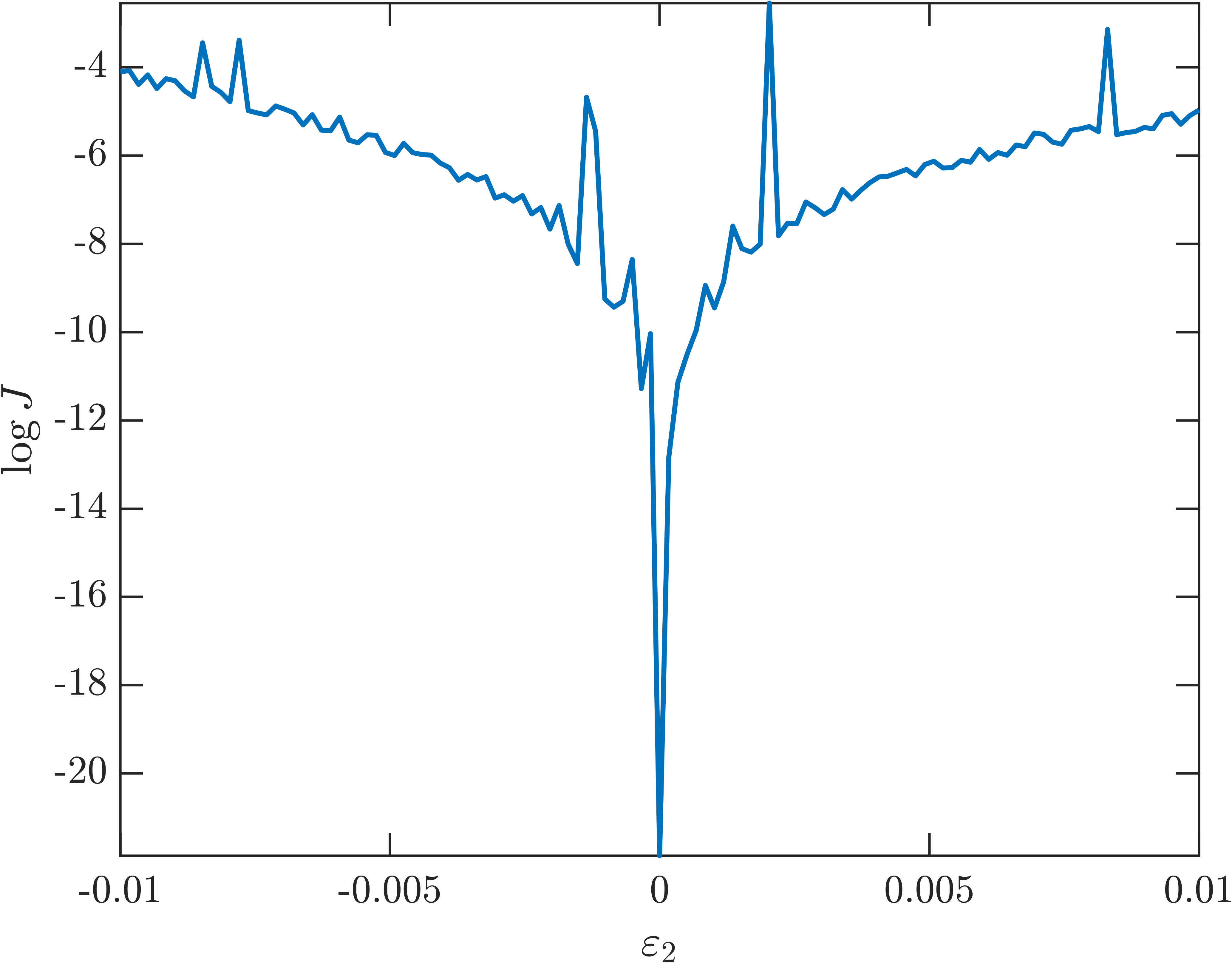}}
	\end{centering}
	\caption{A perturbation study on Problem~\eqref{eq:NLSProb}. We solve Problem~\eqref{eq:NLSProb}, without sparsification, for Hamiltonain densities $h+\varepsilon_1 h_1+\varepsilon_2h_2$, where $h$ is defined in Equation~\eqref{eq:NLSHam} with the Hamiltonians
$h_1=\frac{1}{3}\left|\psi\right|^6$ and $h_2=\frac{1}{2}\left|\partial_x^2 \psi\right|^2$. We see that the loss has a distinguished minimum at the integrable points $\varepsilon_1=0$ (for fixed $\varepsilon_2=0$) and
$\varepsilon_2=0$ (for fixed $\varepsilon_1=0$).}\label{fig:NLSPerturb}
\end{figure}

Just as in previous sections, we aim to interpret our numerical results by introducing sparsification. We omit the numerical values and simply report the mathematical findings that we infer from our computation. These computations typically lead to an unbiased loss of $\mathcal{O}(10^{-15})$. 
That is, we report that with the same sparsification strategies employed earlier, we once again discover a new family of weak Lax pairs. However, this time SILO led us to a Lax pair that simultaneously produces a coupled system of linear and nonlinear Schr\"odinger equations.
\begin{theorem}[Existence of Weak Lax Pairs for the Linear and Nonlinear Schr\"odinger Equations]\label{thm:weakNLS1}
        Let $w(x)=[u(x)\  v(x)]^{\intercal}$ where each component is a complex function in the Schwartz space $\mathcal{S}(\mathbb{R}).$ Then for every $\psi,\varphi\in H^1([0,T];H^3(\mathbb{R}))$, the Lax pair
            $$
L=\left(\begin{array}{cc}
a_1\varphi^* & a_2\psi \\
a_3\psi^* & a_4\varphi
\end{array}\right), \quad P=\left(\begin{array}{cc}
b_1\partial_x^2-b_2|\psi|^2 & 0 \\
0 & -b_1\partial_x^2+b_2|\psi|^2
\end{array}\right)
$$
produces the decoupled pair of equations
\begin{align*}
i \varphi_t&=b_1\varphi_{x x},\\
i \psi_t&=b_1 \psi_{x x}-2b_2|\psi|^2 \psi.
\end{align*}
almost everywhere together with their complex conjugates.

\begin{proof}
    By  direct analogy with the KdV case, an evaluation on $w\in\mathcal{S}(\mathbb{R})^2$ along with integration by parts and function space embeddings wherever necessary reproduces the desired linear and nonlinear Schr\"odinger equations
    (almost everywhere).
\end{proof}
\end{theorem}\

It is interesting to note that this Lax pair not only produces the NLS, but its linear counterpart, too. The free space linear Schr\"odinger equation is, after all, exactly solvable. We note that in the classical case, the Zakharov-Shabat Lax pair produces the NLS and its complex conjugate counterpart (i.e., it does not
reproduce a linear Schr\"odinger model. 

Our numerical findings of sparse Lax pairs helped us in the discovery of the Lax pair in Theorem~\ref{thm:weakNLS1}. During the manual process of generalizing our numerical results to a mathematical statement about Lax pairs, we recognized a degenerate case where the spectrum of $L$ is zero for all time.  This is indeed the sparsest possible weak Lax pair that produces the NLS equation as we show now.
\begin{theorem}[Sparsest Weak Lax Pair for the NLS Equation]\label{thm:weakNLS2}
    Let $w(x)=[u(x)\  v(x)]^{\intercal}$ where each component is a function in the Schwartz space $\mathcal{S}(\mathbb{R}).$ Then for every $\psi\in H^2([0,T];H^3(\mathbb{R})),$ the Lax pair 
    $$
L=\left(\begin{array}{cc}
0 & 0 \\
\alpha\psi^{*} & 0
\end{array}\right), \quad P=\left(\begin{array}{cc}
\beta \partial_x^2+\gamma|\psi|^2 & 0 \\
0 & 0
\end{array}\right)
$$
is the sparsest weak Lax pair that reproduces the NLS equation
$$
i \psi_t=-\beta  \psi_{x x}-\gamma|\psi|^2 \psi
$$
almost everywhere.
\begin{proof}
    In this case, verifying the Lax pair is a short calculation, so we provide the details here. Observe that 
    $$
    \int_{\mathbb{R}}[L,P]wdx=\binom{0}{\int_{\mathbb{R}}\left(\alpha \beta \psi^*u_{xx}+\alpha \gamma|\psi|^2 \psi^* u\right)dx}
    $$
    while
    $$
    \int_{\mathbb{R}}\partial_tLwdx=\binom{0}{\alpha\int_{\mathbb{R}}\psi^*_tudx}.
    $$
    Setting $\int_{\mathbb{R}}\left(\partial_tL-i[L,P]\right)wdx=0$, dividing through by $\alpha,$ and integrating by parts twice, we see that the second component is now
    $$
    \left\langle-i \psi_t^*-\beta  \psi_{x x}^*- \gamma |\psi|^2 \psi^*, u\right\rangle_{L^2(\mathbb{R})}=0,
    $$
which is the desired weak form of the NLS equation. The fact that this is the sparsest Lax pair naturally follows by contradiction.
\end{proof}
\end{theorem}
In a certain sense, the results of SILO in the present
example provided us with only a trivial weak Lax pair that produces the NLS equation. One can argue that this is perhaps not (sufficiently) mathematically useful, but with respect to our computational framework, this is the best we could hope for. It is therefore not surprising that our numerics, aimed at finding the sparsest possible Lax pairs, correctly discover the Lax pairs shown in Theorems~\ref{thm:weakNLS1} and~\ref{thm:weakNLS2}, all of which, under an integral, are sparser than the Zakharov-Shabat one.

\section{Conclusion and Outlook}\label{section:conc}
In many ways, we (hope that we) have demonstrated that SILO is a step in the right direction for automated discovery of integrability in Hamiltonian dynamical systems. Broadly speaking, SILO achieves two major goals; high-precision detection of integrability and the discovery of interpretable, sparse Lax pairs (potentially weak
ones through a relevant definition
that was made precise earlier in the text).

Despite its successes, SILO is certainly not without flaws. SILO consistently finds the six-term lower order Lax pair of Theorem~\ref{thm:strongthm} instead of the typical five-term one reported in the literature. It was only with some guidance that we could reproduce the well-known pair. Additionally, despite attempts to filter out the discovery of weak Lax pairs by sampling from periodic function spaces, SILO, in its current form, did
not succeed in identifying a strong Lax pair for the NLS equation. This warrants future study on how to tune our framework toward a more reliable discovery of $\textit{all}$ Lax pairs that may exist. Additionally, both the function
spaces within which functions are selected and the very choice of
examining the action of the operator equation on the solution
(rather than on arbitrary functions) are aspects worth revisiting
and improving, as may be possible, in future efforts. Indeed, 
accounting in some suitably generalized
way for the operator nature of the Lax equation is a significant
direction that can potentially be targeted for future efforts.

There are numerous additional veins for improvement and continued exploration. One such study is to construct a wide operator hypothesis to attempt to thoroughly investigate the integrability of the Henon-Heiles system. Henon-Heiles indeed has two other known integrable points, yet the Lax pairs known to exist at these points are functionally different, and not just by the placement of the $p$'s and $q$'s, than the one we investigated in this work. Searching over such wide hypotheses is a computational challenge that should be addressed because this is a clear path to the potential discovery of unknown points of integrability over the $(A,B,\varepsilon)$ landscape. There are also opportunities to improve our computational framework to support large-scale nonconvex numerical optimization methods to attack large lattice-based systems such as Ablowitz-Ladik and Toda type systems. There, potentially, further complications
may arise, such as, e.g., the non-standard Poisson brackets
of the Ablowitz-Ladik setting~\cite{kevrekidis2009dnls}.

We also leave behind interesting mathematical questions in the context of the Lax pairs discovered in the PDE setting. Are weak Lax pairs of potential mathematical use, or are they simply weeds in the landscape of viable operator pairs? 
Does the strong Lax pair of Theorem~\ref{thm:strongthm} have an impact on the inverse scattering theory for the KdV equation? 
[A concern in the latter vein is the non-Hermitian nature
of the resulting $L$ operator, especially since
the isospectrality of this operator and the nature
of its eigenvalues are crucial to the entire inverse scattering
machinery~\cite{Ablowitz2011a}.]
Is this strong pair an example of a ``fake" Lax pair~\cite{butler2015two}, that is, one which reproduces the Hamiltonian system yet fails to reproduce the correct conserved quantities? Or does this Lax pair imply Liouville integrability, that is, does it satisfy the geometric conditions of Gui-zhang~\cite{gui1989liouville}?

We believe that sparse symbolic regression techniques will continue to lead to breakthroughs in automated discovery within mathematical physics. Indeed, our intention in the
present work is to plant the seed towards further potential
classical and quantum, low-dimensional and field theoretic
data-driven approaches. Beyond integrability, there is also the question of (maximal) super-integrability on an $N$-degree of freedom Hamiltonian system with $m > N$ integrals up to $m=2N-1$ for maximal super-integrability. Perhaps the best known model, both in the classical and quantum setting, is the Calogero-Moser model \cite{calogero}. Although much is known in the classical setting, in the quantum case, superintegrability is an open question for $N>2$. 

Another excellent testbed for adapting the essence of the SILO framework is the realm of exactly solvable models of statistical and quantum field theories. In fact, in that context, one could think of detecting exact solvability by exploiting the Yang-Baxter relation~\cite{korepin1997quantum}. The 1D and 2D Ising, eight-vertex, quantum Heisenberg, Lieb-Liniger, and Hubbard models are all examples of field theories that are consistent with regard to the Yang-Baxter equation. We ask: Could the sparse identification of Yang-Baxter scattering matrices be the key to automated discovery of integrable quantum field theories as well? 

\section{Acknowledgments}
This material is based upon work supported by the U.S. National Science Foundation under the award PHY-2316622 (JA), PHY-2110030,
PHY-2408988 and DMS-2204702 (P.G.K.) and DMS-2502900 (WZ) and by the Air Force Office of Scientific Research (AFOSR) under Grant No. FA9550-25-1-0079 (WZ). JA and PGK gratefully acknowledge the Society for Industrial and Applied Mathematics (SIAM) postdoctoral support program, established by Martin Golubitsky and Barbara Keyfitz, for helping to make this work possible. JA also acknowledges helpful discussions with Roy Goodman, who suggested studying the Henon-Heiles system, Peter Miller, who pointed out the possibility of fake Lax pairs, and Sebastien Motsch, who suggested we study the effects of using $\mathcal{R}_p$ for $0<p<1$ in sparsification. The authors also acknowledge Nicholas Bagley for his initial efforts on this project.
\bibliography{bibliography}{} 
\bibliographystyle{siam}

\end{document}